\documentclass[11pt,a4paper]{article}
\usepackage{jheppub}
\usepackage{physics}
\usepackage{dsfont}
\usepackage{amsthm}
\usepackage{xcolor}
\usepackage{colortbl}
\usepackage{graphics}
\usepackage{tikz}
\usepackage{url}

\usetikzlibrary{shapes.geometric, arrows}

\tikzstyle{startstop} = [
    rectangle, 
    rounded corners, 
    minimum width=3cm, 
    minimum height=1cm, 
    text centered, 
    draw=black
]

\tikzstyle{io} = [
    trapezium, 
    trapezium left angle=70,
    trapezium right angle=110,
    minimum width=3cm, 
    minimum height=1cm, 
    text centered, 
    text width=2.5cm,
    draw=black
]

\tikzstyle{process} = [
    rectangle, 
    minimum width=3cm, 
    minimum height=1cm, 
    text centered, 
    text width=3cm,
    draw=black
]

\tikzstyle{decision} = [
    diamond, 
    minimum width=3cm, 
    minimum height=1cm, 
    text centered, 
    text width=3cm,
    draw=black
]

\tikzstyle{arrow} = [
    thick,
    ->,
    >=stealth
]

\newcommand{\ra}{\rightarrow}

\newcommand{\ZZ}{{\mathbb Z}}

\newcommand{\frS}{{\mathfrak S}}
\newcommand{\frT}{{\mathfrak T}}

\newcommand{\cB}{\mathcal B}

\newcommand{\cC}{\mathcal C}

\newcommand{\cF}{\mathcal F}

\newcommand{\cS}{\mathcal S}
\newcommand{\cT}{{\mathcal T}}
\newcommand{\cU}{\mathcal U}

\newcommand{\SL}{\mathrm{SL}_2(\ZZ)}
\newcommand{\SLn}{\mathrm{SL}_2(\ZZ_n)}

\newcommand{\hT}{{\hat T}}
\newcommand{\hS}{{\hat S}}

\newcommand{\Ind}{\operatorname{Ind}}
\newcommand{\Gth}{\Gamma_\theta}

\newtheorem{theorem}{Theorem}[section]
\newtheorem{corollary}{Corollary}[theorem]

\title{Classification of Fermionic Topological Orders from Congruence Representations}
\author[a,b,c]{Gil Young Cho,}
\author[a,b]{Hee-Cheol Kim,}
\author[a]{Donghae Seo}
\author[b]{and Minyoung You}

\affiliation[a]{Department of Physics, Pohang University of Science and Technology, Pohang, Gyeongbuk, 37673, Korea}
\affiliation[b]{Asia-Pacific Center for Theoretical Physics, Pohang, Gyeongbuk, 37673, Korea}
\affiliation[c]{Center for Artificial Low Dimensional Electronic Systems, Institute for Basic Science (IBS), Pohang 37673, Korea}

\emailAdd{gilyoungcho@postech.ac.kr}
\emailAdd{heecheol@postech.ac.kr}
\emailAdd{donghae98@postech.ac.kr}
\emailAdd{minyoung.you@apctp.org}

\abstract{The fusion rules and braiding statistics of anyons in \texorpdfstring{$(2+1)$}{21}D fermionic topological orders are characterized by the modular data of a super-modular category. On the other hand, the modular data of a super-modular category form a congruence representation of the \texorpdfstring{$\Gth$}{Gth} subgroup of the modular group \texorpdfstring{$\SL$}{SL}. We provide a method to classify the modular data of super-modular categories by first obtaining the congruence representations of \texorpdfstring{$\Gth$}{Gth} and then building candidate modular data out of those representations. We carry out this classification up to rank $10$. We obtain both unitary and non-unitary modular data, including all previously known unitary modular data, and also discover new classes of modular data of rank $10$. We also determine the central charges of all these modular data, without explicitly computing their modular extensions.
}

\begin{document}

\maketitle
\flushbottom

\section{Introduction}
The emergence of anyons in (2+1)D and their physical properties have been intensely investigated both in condensed matter physics and high-energy physics \cite{Kitaev_2006, MooreSeiberg, WittenJones}. For example, they are at the heart of the physics of important condensed matter systems such as fractional quantum Hall states and spin liquids \cite{FQHE,Wen1989, SpinLiquidReview,Kitaev_2003}. In addition to being of theoretical interest, the physics of anyons is also highly relevant for the realization of the decoherence-tolerant quantum computation, i.e. topological quantum computation \cite{Kitaev_2003, Nayak_QuantumComputation}. In spite of the  strong  interest in the study of anyons, a complete classification of anyon theories has  proven elusive.

Mathematically, a topologically-ordered phase containing anyons is completely characterized (up to invertible topological orders) by the fusion and braiding rules of anyons along with a set of consistency relations; this structure is known as a braided fusion category (BFC) \cite{Wen2015Boson, Wen2015Fermion, BK}. When the fundamental particle in the theory is bosonic, the corresponding mathematical structure is called as the modular tensor category (MTC). Similarly, when the fundamental particle is fermionic, it is called a super-modular category (SMC) \cite{Wen2015Fermion, FermionicModular16}. While such bosonic and fermionic theories are intimately related, they form different mathematical structures --- these will be discussed below in more detail. Thus, the classification programs for bosonic and fermionic theories should be carried out separately.

A direct classification of  the defining data of BFCs is known to be prohibitively difficult because of huge gauge redundancies. Thus, attempts at classification have instead focused on the so-called modular data (MD) of BFCs, which are the ``fingerprints" of  BFCs. The MD are gauge invariant and hence more amenable to classification. See refs. \cite{ClassificationBosonic, ClassificationBosonicbyRank, Mtheoretic, ClassificationRank6, ClassificationRank8} for the previous efforts in this line of thinking. Recently, ref. \cite{Reconstruction} has introduced a method which uses representations of $\SL$ to classify the MD of MTCs, i.e. bosonic topological orders, and used it to classify MD up to rank $6.$ They makes use of the fact that every MD given by a pair $(S,\  T)$ forms a projective \emph{congruence} representation of $\SL$ \cite{Cong_Subgroup_Frob_Schur}. Since every congruence representations of $\SL$ can be constructed explicitly \cite{SymRepSL, SL2Reps}, this gives us a list of candidates from which valid MD can be constructed, and leads to the most complete classification of bosonic topological orders so far obtained. 

There have also been attempts to classify SMCs, which characterize fermionic topological orders. The fusion rules of unitary SMCs have been completely classified up to rank 6 \cite{ClassificationRank6} and partially for rank 8 \cite{ClassificationRank8}, while their explicit MD, have been partially classified in ref. \cite{Wen2015Fermion}. 

We go beyond these results  and obtain a classification of fermionic MD up to rank 10. While we cannot prove that our classification is complete, we recover all previously known unitary MD \cite{Wen2015Fermion, ClassificationRank6, ClassificationRank8}, obtain non-unitary MD (which had previously not been classified), and furthermore discover two completely new classes of MD (with previously unknown fusion rules). The new types of MD are primitive in that they are not obtained from stacking other theories, and they both contain large fusion coefficients ($\hat{N}_{ij}^{ \ \ k} \geq 3$)  and large total quantum dimensions. Moreover, our method allows us to identify the central charge modulo $\frac{1}{2}$ of all these theories, without having to compute their modular extensions. This is in contrast with previous approaches \cite{Wen2015Fermion, FermionicModular16}, where the determination of the central charge requires an explicit knowledge of the modular extensions.

The starting point of our method is the result of ref. \cite{Cong_Subgroup_Super_Modular} that the MD of the fermionic quotient of an SMC form projective congruence representations of $\Gamma_\theta$, the subgroup of $\SL$ generated by $s = \begin{pmatrix}  0 & - 1 \\ 1 & 0 \end{pmatrix}$ and $t^2$ where $ t = \begin{pmatrix} 1 & 1 \\ 0 & 1 \end{pmatrix}$ \cite{FermionicModular16}. Thus, if we have a list of congruence representations of $\Gamma_\theta$, we could hope to do something similar to the bosonic classification of MD carried out in ref. \cite{Reconstruction}. We do exactly that, by first obtaining the list of $\Gamma_\theta$ congruence representations using representation theory, and then constructing and checking potential MD from the representations.

In section \ref{Sec:SMC}, we introduce SMCs and some known facts about them which shall be relevant for our classification procedure. In section \ref{Sec:RepTheory}, we discuss the congruence representation theory of $\Gth$ and prove the necessary theorems which allow us to obtain all congruence representations of $\Gth$. In section \ref{Sec:Construction}, we explain how to construct MD from $\Gth$ representations. In section \ref{Sec:Results} we present our results, which classify fermionic MD up to rank $10$, and compare our results to previous  results. 

\newpage

\begin{table}[t]
    \centering
    \resizebox{\columnwidth}{!}{
    \begin{tabular}{c c }
    
    \multicolumn{2}{c}{\textsc{\Large Notation}} 
    
    \\[5mm]

        $\cB$ &  Super-modular category
        
        \\[2mm]
        $\cB_0$ & Fermionic quotient of $\cB$
        \\[2mm]
        
        $\cC$ & Modular category (may be a spin modular category)
        \\[2mm]
        
        $\cF_0$ &  Trivial super-modular category, also known as $\rm{sVec}$ 
        \\[2mm] 
        
        $f$ & The fermion object of a super-modular or spin modular category
        \\[2mm]
        
        $D$ & Total quantum dimension of a fusion category
        \\[2mm]
        $d_i$ & Quantum dimension of a simple object $i$
        \\[2mm]
        $\theta_i$ & Topological spin of a simple object $i$
        
        \\[2mm]
        $c$ & Chiral central charge of a modular or super-modular category
        
        \\[2mm]
        $N_{ij}^{ \ \ k}$ & Fusion coefficients of a fusion category
        
        \\[2mm]
        $\hat{N}_{ij}^{ \ \ k}$ & Fusion coefficients of the fermionic quotient $\cB_0$
        
        \\[2mm]
        
            $\nu_2(i)$ & Second Frobenius-Schur indicator of a self-dual object $i$
        
        \\[2mm]

        $s,\  t$ & Generators of $\SL$ 
        \\[2mm]
        $S, \ T $ &  $S$- and $T$-matrices of a modular or super-modular category
        
        \\[2mm] 
        $\tilde{\Phi}$ & Projective representation of $\SL$ formed by $S, \ T$ of a modular category $\cC$

        \\[2mm]
        $\Phi$ & Lift of $\tilde{\Phi}$ to a  linear representation of $\SL$
        
        \\[2mm] 
        $ \hS,\ \hT^2$ &  Modular  matrices of a fermionic quotient $\cB_0$
        
        \\[2mm]
        $\tilde{\rho}$ & Projective representation of $\Gth$ formed by $\hS, \ \hT^2$

         \\[2mm]
         $\rho$ & Lift of $\tilde{\rho}$ to a linear representation of $\Gth$
        
          \\[2mm]
        $\mathfrak{S}, \ \mathfrak{T^2}$ & Representation matrices $\rho(s), \ \rho(t^2)$ of $\rho$

        \\[2mm]
        $R$ & Linear representation of $\SL$ induced from $\rho$

        \\[2mm]
        
        $\cS, \ \cT$ & Representation matrices $R(s), \ R(t)$ of $R$

    \end{tabular}}
\end{table}

\newpage

\section{Fermionic topological orders and super-modular categories \label{Sec:SMC}}

Fermionic topological orders are zero temperature phases beyond Landau's symmetry breaking paradigm, realized in a fermionic many-body system \cite{Wen2015Fermion}. In $(2+1)$D,  fermionic topological orders (up to invertible topological orders) are characterized by the fusion rules and braiding statistics of emergent point-like excitations (anyons) together with the fundamental fermionic excitation. The fusion and braiding of these excitations form a categorical structure, known as a SMC  \cite{Wen2015Fermion,FermionicModular16}. 

A SMC is a ribbon fusion category with its nontrivial transparent object  isomorphic to the local fermion object $f$. We refer to ref. \cite{FermionicModular16} for details of SMCs, and will only introduce some key properties which will be necessary for our purposes. Physically, ``transparent" means that $f$ has trivial mutual statistics with any other point-like excitation.  The simple objects of an SMC always come in pairs, as for any anyon $a,$ $f \times a$ is a distinct object. Hence the rank of an SMC is always even. The subcategory of transparent objects (the M{\"u}ger center) of an SMC is $\rm sVec$, the category of super-vector spaces. The trivial SMC is equivalent to $\rm sVec$, and we shall denote it as $\cF_0.$ 

As in bosonic MTCs, an SMC has gauge-invariant data called the $S$- and $T$-matrices, and the anyons form a fusion ring given by 
\begin{equation}
    i \times j =\sum \limits_{k \in \Pi}     N_{ij}^{ \ \ k} k
\end{equation}
where $N_{ij}^{\  \ k}$ are the fusion coefficients and $\Pi$ is the label set of anyons. Other quantities can be written in terms of the $S$- and $T$-matrices:
\begin{itemize}
    \item The topological spin of an anyon $i$, $\theta_i = T_{ii}$
    
    \item The quantum dimension of an anyon $i$, $d_i = \frac{S_{i1}}{S_{11}}$, where the index $1$ corresponds to the vacuum
    
    \item The total quantum dimension $D^2 = \sum_{a \in \Pi} d_i^2 $, where $\Pi$ denotes the label set of the anyons.
\end{itemize}
These data satisfy the \emph{balancing equation} \cite{BK}:
\begin{equation}
    S_{ij} = \frac{1}{D}\theta_i^{-1} \theta_j^{-1} \sum \limits_k N_{ij}^{ \ \ k} \theta_k d_k.
    \label{eq:balancing}
 \end{equation}

The $S$-matrix of an SMC is always degenerate. However, it is known that the MD of an SMC always admit a tensor decomposition \cite{Wen2015Fermion, FermionicModular16}
\begin{equation}
    S = \frac{1}{\sqrt{2}}\begin{pmatrix} 1 & 1 \\ 1& 1\end{pmatrix} \otimes \hS, \quad T = \begin{pmatrix} 1 & 0 \\ 0&-1\end{pmatrix} \otimes \hT.
\label{eq:tensor_matrices}
\end{equation}
where  $\hS$ is unitary. Given $T$, $\hT$ is not well-defined, but $\hT^2$ is. We can always uniquely determine $(\hS, \ \hT^2)$ in terms of $S,\ T$  and vice versa, so we refer to them interchangeably as fermionic MD, and also simply as MD when no confusion with the bosonic case should arise. 

Since the simple objects of an SMC always come in pairs related by fusion with $f$, i.e. $a$ and $a \times f \equiv a^f$, we can decompose the set of simple objects into two. After this decomposition, we obtain the \emph{fermionic quotient} $\cB_0$ of an SMC $\cB$, which is a fusion category with half the number of simple objects as $\cB$. While the decomposition is not canonical, the properties which follow will not depend on the choice \cite{FermionicModular16}. $(\hS$,\ $\hT^2)$ can now be thought of as the MD of the fermionic quotient. The anyons of the fermionic quotient form a fusion ring among themselves, 
\begin{equation}
    i \times j = \sum \limits_{k \in \Pi_0} \hat{N}_{ij}^{ \ \ k} k
\end{equation}
where $\Pi_0$ is the label set of anyons of $\cB_0,$ and $\hat{N}_{ij}^{ \ \ k}$ satisfy
\begin{equation}
    \hat{N}_{ij}^{ \ \ k} = N_{ij}^{ \ \ k} + N_{ij}^{ \ \ k^f}.
    \label{eq:fusion_quotient}
\end{equation}
 
Not all $(\hS, \ \hT)$ can describe a valid SMC. There are  necessary conditions which $(\hS,\ \hT)$ need to satisfy if they are to describe a valid SMC \cite{Wen2015Fermion, FermionicModular16, ClassificationRank6}:

\begin{itemize}
    \item Verlinde formula:
    \begin{equation}
        \hat{N}_{ij}^{ \ \ k} = \frac{2}{D} \sum \limits_{m \in \Pi_0} \frac{\hS_{im} \hS_{jm} \hS_{km}^*}{d_m},
        \label{eq:Verlinde}
    \end{equation}
    
    
    \item Frobenius-Schur indicator condition:
    \begin{equation}
        \pm 1 = \nu_2(a) = \frac{2}{D^2} \sum_{j, k \in \Pi_0} \hat{N}_{jk}^{ \ \ a} d_j d_k \left( \frac{\theta_j}{\theta_k} \right)^2
        \label{eq:FS}
    \end{equation}
    for any self-dual anyon $a$ (i.e. an anyon which satisfies $a \times a = 1 + \ldots$).  Note that while $\theta_i = \hT_{ii}$ depend on the choice of the fermionic quotient, $\hT^2_{ii}$ and hence $\theta_i^2$ are well-defined.
\end{itemize}

As for MTCs, a full characterization of SMCs requires gauge-dependent data called $R$- and $F$-tensors. The MD $(\hS, \ \hT^2)$ are gauge-invariant and much easier to classify, but give only a partial characterization: there may be multiple inequivalent fusion categories with the same MD \cite{MS17}, and even if we find candidate MD which satisfy the conditions \ref{eq:Verlinde} and \ref{eq:FS}, it remains to explicitly construct and prove the existence of a fusion category which gives rise to such MD. However, the MD capture a large part of the physical properties of interest \cite{Wen2015Fermion}, and the conditions \ref{eq:Verlinde} and \ref{eq:FS} are stringent enough that they allow us to narrow down the list of candidates considerably. In this paper, we will refer to MD satisfying the above conditions, plus those discussed in section \ref{Sec:MD_cong} and section \ref{Sec:resolved} (see around eq. \ref{eq:Hmatrix}) as \emph{valid} MD, even though we will not prove the existence of an SMC realizing those MD.

\subsection{Minimal modular extension of super-modular categories}

A \emph{spin modular category} is an MTC with a fermion --- i.e. a simple object $f$ with topological spin $\theta_f = -1$ and $f \times f = 1$ (if there are multiple fermions, we pick a distinguished fermion). Condensing the fermion results in a SMC \cite{FermionicModular16}.  

Given an SMC $\cB$, a \emph{modular extension} of $\cB$ is defined as a spin modular category where a condensation of $f$ yields $\cB.$ The modular extension whose quantum dimension satisfies $D_{\cC}^2 = 2 D^2_{\cB} $ is called a minimal modular extension. Colloquially speaking, a minimal modular extension of an SMC $\cB$ can be thought of as a ``smallest'' possible MTC which contains  $\cB$, built by adding enough anyons with nontrivial braiding with $f$ so that the $S$-matrix becomes non-degenerate. For example, both the toric code MTC and the Ising MTC are minimal modular extensions of the trivial SMC $\cF_0.$ If a minimal modular extension exists, there are always $16$ of them \cite{Lan_Kong_Wen_2016}.

It is a conjecture that every SMC admits a minimal modular extension \cite{FermionicModular16}. We assume this conjecture throughout the paper. Specifically, this assumption is used in section \ref{Sec:lifting} and section \ref{Sec:central_charge}. Alternatively, in case the conjecture fails, our work can be viewed as applying to the subset of  SMCs which do admit a minimal modular extension. As physical topological orders are described by SMCs admitting minimal modular extensions \cite{Wen2015Fermion}, this is at any rate sufficient for physical purposes. We shall also drop the adjective ``minimal'' from here on and simply refer to modular extensions with the understanding that they always refer to minimal modular extensions unless otherwise stated. 

\subsection{Fermionic modular data and congruence representations of \texorpdfstring{$\Gth$}{Gamma\_theta}}
\label{Sec:MD_cong}

Since $S$ is degenerate for SMCs, $S$ and $T$ do not form a projective representation of $\SL$ as in the bosonic case. This seems to make the application of representation theory to classification of SMCs difficult. However, eq. \ref{eq:tensor_matrices} allows us to define $\hS$ and $\hT^2$, which are unitary matrices. These together generate a projective representation of a subgroup $\Gth$ of $\SL$ \cite{Wen2015Fermion, FermionicModular16}. The fact only $\hT^2$ is well-defined reflects the fact that $\Gth$ is a  subgroup of $\SL$ generated by $s$ and $t^2$. 

When the SMC is \emph{split}, i.e. it can be written as $\mathcal{C}_b \boxtimes \cF_0$ for some bosonic MTC $\mathcal{C}_b$, we will be able to find a choice $\hT$ of the square root of $\hT^2$ such that 
\begin{equation}
    (\hS \hT)^3 = \varphi \hS^2
    \label{eq:proj_extendalbe}
\end{equation}
(where $\varphi$ is some phase). In this case, the projective $\Gth$-representation formed by $(\hS, \ \hT^2)$ will be called \emph{projectively extendable}, as it can be extended to a projecetive representation $(\hS, \ \hT)$ of $\SL$. However, for those SMCs which are not split, no choice of $\hT$ will be able to satisfy eq. \ref{eq:proj_extendalbe} \cite{FermionicModular16, Cong_Subgroup_Super_Modular}. For example, the following $2$-dimensional matrices, corresponding to a rank-$4$  SMC denoted ${\rm PSU}(2)_6$, are not projectively extendable:
\begin{equation}
    \hat{S} = \frac{1}{\sqrt{4 + 2 \sqrt{2}}} \begin{pmatrix}1 & 1 + \sqrt{2}\\ 1+\sqrt{2} & -1\end{pmatrix}, \  \hat{T}^2 = \begin{pmatrix}1 &0 \\ 0& -1\end{pmatrix}.
\end{equation}
The SMC ${\rm PSU}(2)_6$ is an example of a \emph{primitive} SMC: one which cannot be written as a stacking of two other fusion categories. Primitive SMCs, together with primitive bosonic MTCs, can be used to build all SMCs via stacking. Not all non-split SMCs are primitive: the stacking of two primitive SMCs is not primitive but not necessarily split. 

In ref. \cite{Cong_Subgroup_Super_Modular}, it was shown that any projective $\Gth$-representation arising from an SMC (assuming that the SMC admits a modular extension) is a \emph{congruence} representation, i.e. its kernel contains a principal congruence subgroup of $\SL$. The definition and properties of congruence representations of $\SL$ and $\Gth$ will be detailed in section \ref{Sec:congruence_rep_theory}. In the sequel, any representation of $\SL$ or $\Gth$ we mention will be assumed to be congruence unless otherwise stated.


\section{Classification of congruence representations of \texorpdfstring{$\Gth$}{Gamma\_theta}}
\label{Sec:RepTheory}

A complete list of congruence irreps of $\SL$, organized either by level or by dimension, can be obtained from ref. \cite{SL2Reps}.
In this section we will explain how to obtain the congruence representations of $\Gth$ from the congruence representations of $\SL.$  The essential idea is the theory of induction and restriction of representations, which allows us to obtain the irreps of a subgroup $H < G$ if we know the irreps of  $G$. 

\subsection{Congruence representations of  \texorpdfstring{$\SL$}{SL} and  \texorpdfstring{$\Gth$}{theta}  }
\label{Sec:congruence_rep_theory}
We first review the definition and properties of congruence representations. 

Recall that $\SL$ is the group of $2 \times 2$ matrices of unit determinant. A \emph{principal congruence subgroup of level $n$} (for some positive integer $n$) is defined as a subgroup
\begin{equation}
    \Gamma(n) = \Big\{ \begin{pmatrix} a & b \\ c & d 
    \end{pmatrix} \in \SL: ad \equiv 1 \mod n, \ bc \equiv 0 \mod n \Big\}.
\end{equation}
$\Gamma(n)$ becomes smaller as $n$ gets larger, and in fact, $\Gamma(n_1)$ is a subgroup of $\Gamma(n_2)$ if $n_2$ is divisible by $n_1$.

A congruence representation $R$ of $\SL$ is a representation $R: \SL \rightarrow \mathrm{ GL} (V) $ whose kernel $\ker R$ contains $\Gamma(n)$ as a subgroup. The smallest $n$ for which $R$ satisfies $\ker R \geq \Gamma(n)$ is called the \emph{level} of $R.$ Equivalently, a congruence representation of $\SL$ is a representation of $\SL$ which factors through the finite group $\SLn$. Note that $\SLn \simeq \SL/\Gamma(n).$

Similarly, for $\Gth < \SL$, we define congruence representations as those representations $\rho$ of $\Gth$ which satisfy $\ker \rho \geq \Gamma(n)$, or equivalently, those representations of $\Gth$ which factor through the finite group $\Gth/\Gamma(n).$ The smallest $n$ for which the preceding condition is satisfied will be the level of $\rho$. Note that, $\Gth$ itself contains $\Gamma(2)$ as a subgroup, while it does not contain any $\Gamma(k)$ for $k$ odd. Hence the level of a congruence representation of $\Gth$ is always even. 

In terms of generators, $\SL$ can be presented as 
\begin{equation}
    \langle s, t: s^4 = 1, (st)^3 = s^2 \rangle,
\end{equation}
whereas $\Gth$ is presented as
\begin{equation}
    \langle s, t^2: s^4 = 1, s^2 t^2 = t^2 s^2 \rangle. 
\end{equation}
The relations of $\Gth$ are too simple to admit an interesting representation theory. However, if we require the representations to be \emph{congruence}, there are many more relations between the generators. These relations depend on the level $n$. For congruence representations of $\SL$, the relations are presented in, e.g. \cite{Eholzer}. For congruence representations of $\Gth$, we present the explicit relations in appendix \ref{Sec:Cong}.

\subsection{Representations of a subgroup}

Consider a finite group $G$ and a subgroup $H < G$. Suppose we have a representation of $G$, i.e. some $R: G \ra \mathrm{GL}(V)$ where $V$ is some vector space.  We can obtain a representation $R|_{H}$ of $H$ by \emph{restriction}, which simply means that we limit ourselves to $R(h)$ such that $h \in H < G$.  If the restriction $R|_H$ of an irrep $R$ is again irreducible, both $R$ and $R|_H$ are of the same dimension and we say that $R|_H$ \emph{extendable}, since it is an $H$-representation that can be extended to a $G$-representation by simply assigning values to those elements of $G$ not in $H$. In general, however, irreps do not map to irreps under restriction.

On the other hand, given any representation $\pi$ of $H$, we can construct an  \emph{induced representation}  ${\rm Ind}^G_H \pi$ of $G$ (this is unique for a given $\pi$). While not every $H$-representation can be \emph{extended} to a $G$-representation, every $H$-representation can be \emph{induced} to a $G$-representation. Restriction and induction are ``adjoint'' to each other in the following sense (this property is known as \emph{Frobenius reciprocity}):

\begin{theorem}[Frobenius reciprocity]
Given $H < G$ and an irrep $\pi$ of $H$, the induced representation of $\pi$ decomposes as a direct sum of $G$-irreps $R_i$, where each irrep appears with the multiplicity $m_i$ equal to the number of times its restriction to $H$ contains $\pi$. In other words, $\Ind^G_H \pi = \bigoplus_i m_i R_i$ where $R_i$ are irreps of $G$ such that $\Res_H R_i = \bigoplus_i m_i \pi \oplus ...$ ($...$ are other irreps).
\end{theorem}
Frobenius reciprocity allows us to obtain every irrep of $H$ from restriction of irreps of $G$:
\begin{corollary}
Given $H \subset G$, every irrep $\pi$ of $H$ is contained in the restriction of some irrep $R$ of $G$, i.e. $\Res_H R = m \pi \oplus ...$ for some $R,$ where $m$ is the multiplicity.  
\end{corollary}
\begin{proof}
Every irrep $\pi$ of $H$ has an induced representation $\Ind_H^G \pi,$ which is a representation of $G,$ and this decomposes as $\bigoplus_i m_i R_i$ for irreps $R_i$ of $G$. By Frobenius reciprocity, $\Res_H R_i = m_i \pi \oplus ...$ so these $R_i$ are precisely those irreps of $G$ whose restriction to $H$ contain $\pi.$
\end{proof}

An explicit description of induced representations is given as follows.  Let $\pi : H \ra {\rm GL}(V)$ be a representation of $H < G$ on a vector space $V.$ Let $N = [H: G] $ be the index of $H$ in $G$ and $g_i,$ $i = 1, ..., N$, be the full set of representatives of left cosets in $G/H$.  The $G$-representation ${\rm Ind}^G_H \pi$ acts on the vector space 
$$W = \bigoplus \limits_{i = 1}^N g_i V, $$
i.e. $N$ copies of $V$. For any $g \in G$, its action on this vector space is given by the following:
\begin{itemize}
    \item First, for each coset $g_i$, $g \cdot g_i = g_{j(i)} h_i$ for some (possibly different) coset corresponding to $g_{j(i)}$ and $h_i \in H$. Once we have fixed the set  $\{ g_i \}$  of coset representatives, this decomposition is unique.
    
    \item The $g$-action permutes the cosets according to $g_i \mapsto g_{j(i)}$ .
    
    \item Moreover, on each subspace $g_i V$, $h_i$ acts by $\pi(h_i).$
\end{itemize}

\subsection{Restriction and induction for representations of \texorpdfstring{$\Gth \subset \SL$}{Gamma\_theta in SL2Z}}

How do these results apply to the case at hand? Both $\SL$ and $\Gth$ are infinite, non-compact groups. However, $\SLn$ and $\Gth/\Gamma(n)$ are finite groups for every (even) $n,$ and $\Gth/\Gamma(n) < \SLn$.

Given a $\SL$-representation $R$ of level $n$, we can restrict it to $\Gth$ straightforwardly. Denote the restricted representation by $R|_{\Gth}$. Since $\ker R \geq \Gamma(n)$, $\ker  R_{\Gth} \geq \Gth \cap \Gamma(n).$ Since $\Gamma(n) \subset \Gth$ for every even $n,$ $\Gth \cap \Gamma(n) = \Gamma(n)$ if $n$ is even and $\Gamma(2n)$ if $n$ is odd.  On the other hand, $\ker R|_{\Gth}$ cannot contain $\Gamma(n')$ for $n' < n$: if it were the case, we could think of $R|_{\Gth}$ as a representation of $\Gth/\Gamma(n')$ and induce it to a representation of ${\rm SL}(\ZZ_{n'})$. Since the induced representation contains $R$ by Frobenius reciprocity, $\ker R > \Gamma(n')$, but this contradicts the fact that $R$ is of level $n.$ Hence the level of $R|_{\Gth}$ is $n$ (if $n$ is even) or $2n$ (if $n$ is odd). 

Thus, every irrep of $\Gth/\Gamma(n)$ (where $n$ is always even) can be obtained from the decomposition into irreps of the restriction of irreps of $\SL$ and ${\rm SL} (\ZZ_{n/2}).$  In other words, every congruence irrep of $\Gth$ can be obtained from restricting and decomposing the congruence irreps of $\SL$. This is the key result which enables us to obtain the full list of congruence representations of $\Gth$ (up to a given dimension).

In order to facilitate the computation, we also make explicit use of induction. For simplicity, we shall speak of the induction from $\Gth$ to $\SL$ in the sequel, but technically this should always be understood as an induction from $\Gth/\Gamma(n)$ to $\SLn$, which are both finite groups. Note that $[\Gamma_{\theta}: \SL  ] = [\Gth/\Gamma(n):\SLn] = 3$ \cite{sl2zConrad}, so that if we start with a $d$-dimensional representation of $\Gth/\Gamma(n)$, the dimension of the induced representation is always $3d$. 


 A choice of left coset representatives of $\Gth \subset \SL$ is given by $1, t, st $ . Let us denote the $\Gth$-representation by $\rho: \Gth \ra {\rm GL}(V)$, and its induced representation by $R: \SL \ra {\rm GL} (W)$. If the action of $\rho$ on $V$ is represented by the matrices $\rho(s) = \frS$ and $\rho(t^2) = \frT^2$, then the action of $R$ on $W = V \oplus s V \oplus st V$ is given by the matrices $R(s) = \cS$ and $R(t) = \cT$ which take the block form
\begin{equation}
    \cS =  \begin{pmatrix} \frS & 0 &  0  \\ 0 & 0 & \frT \frS^2  \\ 0 & \mathds{1} & 0   \end{pmatrix}, \ \cT = \begin{pmatrix}0 & \frT^2 & 0 \\ \mathds{1} & 0 & 0   \\ 0 & 0 &  (\frS \frT^2)^{-1} \end{pmatrix}.
\end{equation}
This explicit form of the induced representation allows us to efficiently compute  $\Gth$-irreps of dimension $d$ coming from $\SL$-irreps of dimension $3d$, using the reverse induction formula detailed in appendix \ref{Sec:reverse_induction}.

The restriction to $\Gth$ means the generators are now 
\begin{equation}
  \cS =  \begin{pmatrix} \frS & 0 &  0  \\ 0 & 0 & \frS^2  \\ 0 & \mathds{1} & 0   \end{pmatrix}, \ \cT^2 = \begin{pmatrix} \frT^2 & 0 & 0 \\ 0 & \frT^2 & 0   \\ 0 & 0 &  \left((\frS \frT)^2 \right)^{-2} \end{pmatrix}.   
\end{equation}
Note that the restriction of the induced representation indeed contains the original $\Gth$-representation in the first block.

\subsection{Constraints from Frobenius reciprocity}

Frobenius reciprocity heavily constrains which $\SL$ irreps can give rise to a $\Gth$ irrep of a given dimension $d.$ Consider a $\Gth$ irrep $\rho$ of dimension $d.$ Its induce representations takes the form
\begin{equation}
    \operatorname{Ind} \rho = \bigoplus_i m_i R_i
\end{equation}
for some $\SL$ irreps $R_i$, which in turn satisfy
\begin{equation}
    \operatorname{Res} R_i = m_i \rho \oplus ...
\end{equation}
Since $\dim \operatorname{Ind} \rho = 3d$, $\operatorname{Ind}\rho$ (if it is not a $3d$-dimensional irrep) can only decompose as $3d = d + d + d$ (in which case $\rho$ is extendable), or $(d+a) + (2d-a)$ for some $0 \leq a  \leq  d.$ $a = d$ or $a = 0$ corresponds to the decomposition $d + 2d $, which makes $\rho$ extendable.  Since we are interested in non-extendable representations, we take the decomposition $3d = (d+a) + (2d - a)$ for $0< a < d.$ The $(d+a)$-dimensional irrep $R_{d+a}$ satisfies $\operatorname{Res} R_{d+a} = \rho \oplus \sigma \oplus \ldots $ where $\rho$ is $d$-dimensional, so $\sigma$ can  at most be  $a$-dimensional. $\operatorname{Ind} \sigma $ should in turn contain $R_{d+a}$, so we need 
\begin{equation}
    3a \geq \dim \operatorname{Ind} \sigma \geq \dim R_{d+a} = d+a.
\end{equation}
This translates to $ a \geq \frac{d}{2}$, or, 
\begin{equation}
    d +a \geq \frac{3}{2}d.
    \end{equation}
On the other hand, if $a > \frac{d}{2}$, $2d - a < \frac{3}{2}d$, so the $(2d-a)$-dimensional irrep $R_{2d-a}$ would not satisfy the above requirements. Hence we need exactly $a = \frac{d}{2}$, or, in other words,  
\begin{equation}
    d + a = \frac{3}{2}d.
\end{equation}

Thus, for any given $\Gth$-irrep $\rho$, either its induced representation is an $\SL$-irrep of dimension $3d$, or it decomposes into two $\SL$-irreps of dimension $\frac{3}{2}d$ each. When $d$ is odd, the latter possibility is precluded  as $\frac{3}{2}d$ is not an integer.

For example, if we are interested in obtaining $4$-dimensional $\Gth$-irreps, we should look at the restrictions of 
\begin{itemize}
    \item $4$-dimensional $\SL$-irreps (these give rise to extendable irreps)
    \item $6$-dimensional $\SL$-irreps
    \item $12$-dimensioanl $\SL$-irreps
\end{itemize}
and obtain their irreducible components.

\subsection{Computation of \texorpdfstring{$\Gth$}{Gth} irreps }

\begin{figure}
    \centering
    \begin{tikzpicture}[node distance=2cm, scale=0.5]
    
    \node (start) [startstop] {Start};
    
    \node (in1) [io, below of=start, xshift=-5cm] {$3d$-dim $\SLn$ irrep};
    
    \node (pro1) [process, below of=in1, yshift=-1.3cm] {Reverse induction};
    
    \node (in2) [io, below of=start] {$\frac{3}{2}d$-dim $\SLn$ irrep};
    
    \node (res) [process, below of=in2] {Restriction};
    
    \node (pro2) [process, below of=res] {Simultaneous BD};
    
    \node (in3) [io, below of=start, xshift=5cm] {$d$-dim $\SLn$ irrep};
    
    \node (res2) [process, below of=in3, yshift=-1.3cm] {Restriction};
    
    \node (dec1) [decision, below of=pro2, yshift=-1.2cm] {Cong. irrep?};
    
    \node (out1) [io, below of=dec1, yshift=-1.2cm] {$d$-dim $\Gth$ irrep};
    
    \node (stop) [startstop, below of=out1] {Stop};
    
    
    \draw [arrow] (start) -- (in1);
    
    \draw [arrow] (start) -- (in2);
    
    \draw [arrow] (start) -- (in3);
    
    \draw [arrow] (in1) -- (pro1);
    
    \draw [arrow] (pro1) -- (dec1);
    
    \draw [arrow] (in2) -- (res);
    
    \draw [arrow] (res) -- (pro2);
    
    \draw [arrow] (pro2) -- (dec1);
    
    \draw [arrow] (in3) -- (res2);
    
    \draw [arrow] (res2) -- (dec1);
    
    \draw [arrow] (dec1) -- node[anchor=east] {yes} (out1);
    
    \draw [arrow] (out1) -- (stop);
    
    \draw [arrow] (dec1.east) -- node[anchor=south] {no} ++(4cm,0) |- (stop);
    
    \end{tikzpicture}
    \caption{Flowchart for computation of $\Gth$ irreps.}
    \label{fig:GthRepFlowchart}
\end{figure}
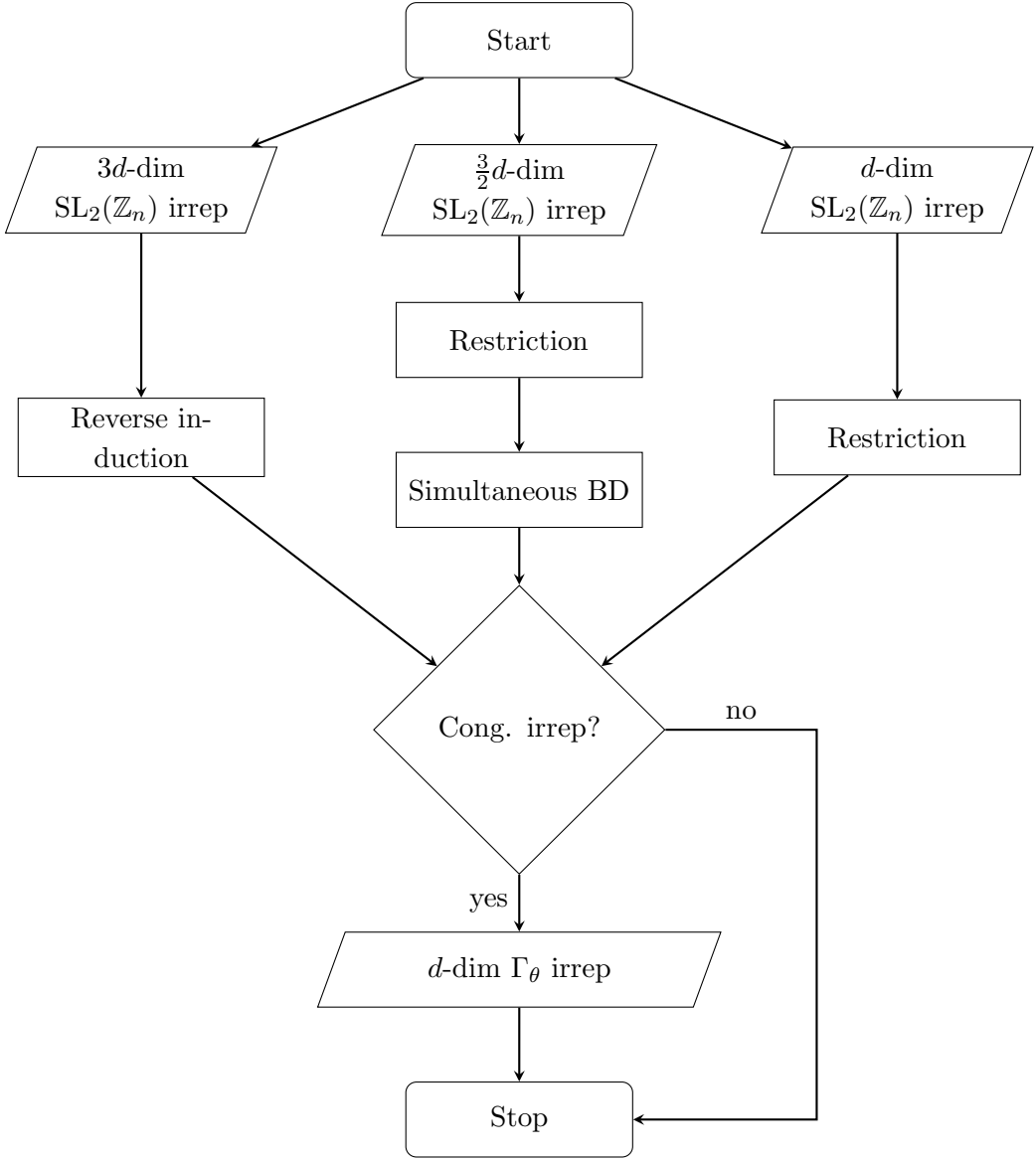

For a given dimension $d$, we look at $\frac{3}{2}d$-dimensional $\SL$ congruence irreps and $3d$-dimensional $\SL$ congruence irreps. In the case of $\frac{3}{2}d$-dimensional $\SL$-irreps, we directly restrict to $\Gth$, which in practice means looking at $T^2$ instead of $T$. Since $T^2$ is given in diagonal form, we apply an appropriate orthogonal transformation to the degenerate eigenspaces of $T^2$ in order to block-diagonalize $S$. If the block-diagonalization is successful, we obtain a $d$-dimensional irreducible component, which will be a $d$-dimensional congruence  irrep of $\Gth$.

For $3d$-dimensional $\SL$-irreps, we use the reverse induction formula, detailed in appendix \ref{Sec:reverse_induction}, to obtain a list of candidates. Then we check the congruence conditions, appendix \ref{Sec:Cong}, to obtain a final list of congruence irreps.

We compute all congruence irreps of $\Gth$ up to dimension $5$. The flowchart for computation of $\Gth$ irreps is shown in figure \ref{fig:GthRepFlowchart}.


\section{Construction of modular data from congruence representations of \texorpdfstring{$\Gth$}{Gamma\_theta}}
\label{Sec:Construction}

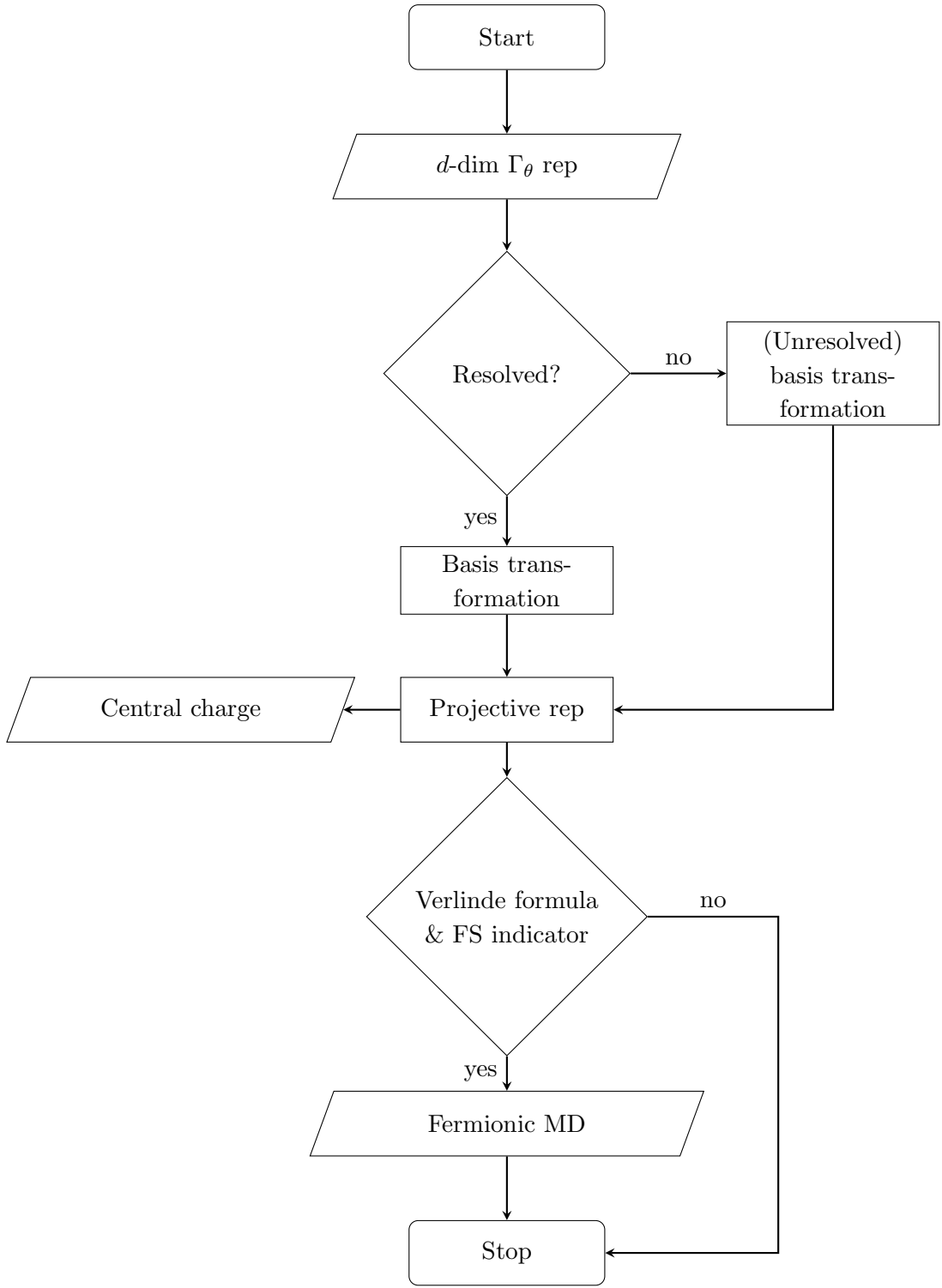
\begin{figure}
    \centering
    \begin{tikzpicture}[node distance=2cm, scale=0.5]
    \node (start) [startstop] {Start};
    
    \node (in) [io, below of=start] {$d$-dim $\Gth$ rep};
    
    \node (dec1) [decision, below of=in, yshift=-1.2cm] {Resolved?};
    
    \node (pro1) [process, below of=dec1, yshift=-1.2cm] {Basis transformation};
    
    \node (pro2) [process, below of=pro1] {Projective rep};
    
    \node (out1) [io, left of=pro2, xshift=-3cm] {Central charge};
    
    \node (dec2) [decision, below of=pro2, yshift=-1.2cm] {Verlinde formula \& FS indicator};
    
    \node (out2) [io, below of=dec2, yshift=-1.2cm] {Fermionic MD};
    
    \node (pro3) [process, right of=dec1, xshift=3cm] {(Unresolved) basis transformation};
    
    \node (stop) [startstop, below of=out2] {Stop};
    
    
    \draw[arrow] (start) -- (in);
    
    \draw[arrow] (in) -- (dec1);
    
    \draw[arrow] (dec1) -- node[anchor=east] {yes} (pro1);
    
    \draw[arrow] (dec1) -- node[anchor=south] {no} (pro3);
    
    \draw[arrow] (pro3) |- (pro2); 
    
    \draw[arrow] (pro1) -- (pro2);
    
    \draw[arrow] (pro2) -- (out1);
    
    \draw[arrow] (pro2) -- (dec2);
    
    \draw[arrow] (dec2) -- node[anchor=east] {yes} (out2);
    
    \draw[arrow] (out2) -- (stop);
    
    \draw[arrow] (dec2.east) -- node[anchor=south] {no} ++(4cm,0) |- (stop);
    
    \end{tikzpicture}
    \caption{Flowchart for construction of fermionic MD.}
    \label{fig:MDFlowchart}
\end{figure}

Once we have obtained a list of congruence irreps of $\Gth$ from congruence irreps of $\SL$, using methods outlined in section \ref{Sec:RepTheory}, we can use them to construct candidate MD. As in the bosonic case  \cite{Reconstruction}, we first  construct a $\Gth$-representation $\rho_{\rm isum}$ for a given dimension $d$ as a direct sum of $\Gth$-congruence irreps. Then, by applying an orthogonal matrix $U$ and a signed diagonal matrix $V$, we put it into a specific basis which makes it a candidate for a MD. After obtaining the list of candidates, we can check the necessary conditions for being a valid MD of an SMC, such as the Verlinde formula (eq.  \ref{eq:Verlinde}) and the Frobenius-Schur indicator condition (eq. \ref{eq:FS}). In this paper, we carry this program out up to dimension $5$ (which corresponds to SMCs of rank $10$). The flowchart for the the construction of fermionic MD from $\Gth$-representations is shown in figure \ref{fig:MDFlowchart}.

\subsection{Basis transformation and resolved representations}
\label{Sec:resolved}
The MD $(\hS, \ \hT^2)$ are basis-dependent quantities, and even if $(\hS, \ \hT^2)$ form a reducible representation, the corresponding modular category may be indecomposable. Thus, if we are interested in fermionic  modular  matrices $\hS, \ \hT^2$ of dimension $d$, we need to look at all $\Gth$-representations of dimension $d$ (including reducible ones), in all possible valid bases. Following ref. \cite{Reconstruction}, we denote by $\rho_{\rm isum}$ the direct sum of irreps in the basis coming from our list of symmetric $\Gth$-irreps (in their case $\rho_{\rm isum}$ denotes a direct sum of  $\SL$-irreps), and $\rho = U \rho_{\rm isum} U^{-1}$ the basis-changed version, which is a candidate for the MD (more precisely, $\rho$ will be a linear lift of the projective representation formed by $(\hS, \ \hT^2$). 

There are several conditions for a valid basis. 
\begin{itemize}
    \item The first is that $\rho(t^2)$ should be diagonal. As our $\Gth$-irrep data is all in this form, this condition is automatically satisfied by any $\rho_{\rm isum}$. In order to preserve this under a transformation $U \rho_{\rm isum} U^{-1}$, $U$ can only act block-diagonally, where each block corresponds to a degenerate subspace of the eigenvalues of $\rho(t^2).$

\item We also require that $\rho(s)$ is symmetric. As $\rho_{\rm isum}$ is always symmetric, we need $U$ to be a combination of signed diagonal matrices $V$ and an orthogonal matrix $U_0$ (where the orthogonal matrix will be block-diagonal, by the above condition), $U = V U_0.$ 

\item Another condition is that $\rho(s)$ should not have zeroes in the first row (or, equivalently, the column, since it is symmetric), i.e. $\rho(s)_{1 i } \neq 0 \  \forall \ i$, corresponding to the fact that quantum dimensions cannot be zero. This leads to the  $t^2$-spectrum  condition, which states that whenever $\rho_{\rm isum}$ is a direct sum, the $t^2$-spectrum (the set of eigevnalues of $\rho(t^2)_{\rm summand}$) of each direct summand should have nonempty overlap. 
See Theorem 3.18 of \cite{MTCbyRank2015} for a proof. (They deal with $\SL$-representations and hence the $t$-spectrum, rather than the $t^2$-spectrum, but the idea is the same.) 
\end{itemize}

Accordingly, for a given dimension, we build direct sums of $\Gth$-irrpes and organize them into types, according to how much overlap their $t$-spectra have. For example, for dimension $5$, we consider the following types of representations $\rho_{\rm isum}$:
\begin{itemize}
    \item $5$d irreps
    \item $4+1$ Type (2)
    \item $3+2$ Type (2)
    \item $3+2$ Type (2, 2)
    \item $3+1+1$ Type (2, 2)
    \item $3+1+1$ Type (3)
    \item $2+2+1$ Type (3)
    \item $2+2+1$ Type (3, 2)
    \item $2+1+1+1$ Type (4)
\end{itemize}
where Type $(a,\ b)$ denotes that the eigenvalues of $\rho(t^2)$ overlap in sets of sizes $a$ and $b$. For example, if the five eigenvalues are $\{ e^{2 \pi i \frac{1}{3}}, e^{2 \pi i \frac{1}{3}}, e^{2 \pi i \frac{1}{3}}, e^{2 \pi i \frac{1}{8}}, e^{2 \pi i \frac{1}{8}} \}$, the representation is of Type (3, 2). In dimension $5,$ $1+1+1+1+1$ Type (5) does not yield any valid MD. We prove in appendix \ref{Sec:5overlaps} that a direct sum of $1$-dimensional representations can only give rise to split SMCs.

When $\rho_{\rm isum}$ is irreducible  the $t^2$-spectrum is non-degenerate (this holds at least up to dimension $5$), so there is no further possibility of orthogonal transformation $U_0$ available. In such a case we simply perform all possible signed diagonal transformations, $V \rho_{\rm isum} V^{-1}$, which gives us the candidate $\rho$.

When $\rho_{\rm isum}$ is a direct sum, the $t^2$-spectrum is degenerate, and in each degenerate eigenspace of dimension $d_\theta$ (corresponding to the topological spin $\theta$) we can perform an orthogonal transformation of dimension $d_\theta$. 

The possible orthogonal transformations are in fact heavily constrained for the so-called `resolved representations' \cite{Reconstruction}, for which the degenerate eigenspace can be `resolved' (i.e. the degeneracy lifted) by the set of matrices 
\begin{equation}
    H(a) \equiv \rho(s)^2 \rho(t)^{-(a-1)} \rho(s) (\rho(t)^2 \rho(s))^{\bar{a} - 1} \rho(t)^{-(a-1)} \rho(s)
    \label{eq:Hmatrix}
\end{equation}
where $a$ is an element of $\ZZ_n^\times$ which satisfies $\theta^{a^2} = \theta$. Here, $n$ is the level of the representation $\rho$, and $\bar{a}$ is the inverse of $a$ modulo $n$, i.e. $a \bar{a} \equiv 1 \mod n.$ Due to a theorem related to Galois conjugation \cite{Cong_CFT}, each $H(a)$ should be a signed permutation for a valid modular representation $\rho$. The theorem is proved for the bosonic case, but assuming the existence of a minimal modular extension, we can prove the same for fermionic theories, and hence this also applies to our case. The fact that $H(a)$ has to be a signed permutation matrix after the orthogonal transformation $U_0$ places severe constraints on what $U_0$ can be for resolved representations. See section C.1 of ref. \cite{Reconstruction}. We apply their logic to the fermionic case, and find that we need only consider the following orthogonal transformations for the resolved degenerate eigenspaces (as mentioned above and in appendix \ref{Sec:5overlaps}, we need not consider a $5$-dimensional eigenspace):
\begin{itemize}
        \item For a $2$-dimensional subspace,
        \begin{equation}
        \begin{pmatrix} \cos \phi & -\sin \phi \\ \sin \phi & \cos \phi \end{pmatrix}
    \end{equation}
    with $\phi = 0, \  \pi/4, \ -\pi/4$.
    
    \item For a $3$-dimensional subspace,     \begin{equation}
        \begin{pmatrix} \cos \phi & -\sin \phi  & 0 \\ \sin \phi & \cos \phi  & 0 \\ 0 & 0 & 1\end{pmatrix}, \ 
        \begin{pmatrix} \cos \phi &0 &  -\sin \phi \\ 0 & 1 & 0 \\\sin \phi & 0 & \cos \phi \end{pmatrix}, \ 
        \begin{pmatrix} 1 & 0 & 0 \\ 0& \cos \phi & -\sin \phi \\ 0& \sin \phi & \cos \phi \end{pmatrix} 
    \end{equation}
with $\phi = 0, \ \pi/4, \ -\pi/4.$

\item For a $4$-dimensional subspace,
      \begin{align}
            &\begin{pmatrix} \cos \phi & -\sin \phi  & 0  & 0 \\ \sin \phi & \cos \phi  & 0 & 0 \\ 0 & 0 & 1 & 0 \\ 0 & 0& 0& 1 
        \end{pmatrix},  \ 
        \begin{pmatrix} \cos \phi &0 &  -\sin \phi & 0  \\ 0 & 1 & 0 & 0 \\\sin \phi & 0 & \cos \phi & 0 &  \\ 0 & 0 & 0& 1 \end{pmatrix},  \ 
        \begin{pmatrix} 1 & 0 & 0 & 0 \\ 0& \cos \phi & -\sin \phi & 0  \\ 0& \sin \phi & \cos \phi  & 0 \\ 0 & 0 & 0& 1\end{pmatrix}, \nonumber \\  
        &\begin{pmatrix} \cos \phi & 0 & 0 & -\sin \phi   \\ 0& 1 & 0 & 0 \\ 0& 0& 1 & 0 \\ \sin \phi & 0 & 0 & \cos \phi \end{pmatrix},  \ 
        \begin{pmatrix} 1 & 0 & 0 & 0 \\ 0& \cos \phi &0 &  -\sin \phi \\ 0 & 0& 1 & 0 \\0& \sin \phi & 0 & \cos \phi \end{pmatrix}, 
        \ \begin{pmatrix} 1 & 0 & 0 & 0 \\ 0& 1 & 0 & 0 \\ 0& 0& \cos \phi & -\sin \phi \\ 0& 0& \sin \phi & \cos \phi \end{pmatrix} 
          \end{align}
          with $\phi = 0, \pi/4, -\pi/4.$
\end{itemize}

Hence, for resolved representations, there are only a discrete set of possible candidates for MD. The vast majority of known valid MD come from resolved representations -- in fact, up to rank $8$, for which there is a more or less complete classification for unitary SMCs, all but one of them come from resolved representations (the one exception corresponds to the toric code stacked with the trivial fermionic theory, $4^B_0 \boxtimes \cF_0$). For rank $10$, a few of the known unitary SMCs are obtained from unresolved representations. We discuss how we obtained them, as well as their non-unitary versions, and our general (though incomplete) strategy for dealing with unresolved representations, in section \ref{Sec:unres}.


\subsection{From linear to projective representations and modular data}
\label{Sec:lifting}

Once we obtain the candidate linear representations $\rho = U \rho_{\rm isum} U^{-1},$ we can easily construct the MD by
\begin{equation}
    \hS = \frac{|\rho(s)_{11}|}{\rho(s)_{11}} \rho(s), \ \hT^2 = \frac{\rho(t^2)}{\rho(t^2)_{11}}
    \label{eq:projectivize}
\end{equation}
where the vacuum corresponds to the first index. $\hS, \ \hT^2$ now only satisfy the relations of the congruence representation projectively, and the level may change. 

$\rho$ may be thought of as a lift of the projective representation formed by $\hS, \ \hT^2$ to a linear representation. If every projective representation of $\Gth$ formed by MD admits such a linear lift, then we can claim that our search for MD is complete, since we begin with a complete list of linear representations of a given dimension. For $\SL$, the existence of linear lifts of the projective representations formed by bosonic MD is guaranteed by Theorem II of ref. \cite{Cong_CFT}. We now prove a similar theorem for fermionic MD and $\Gth$.
\begin{theorem}
\label{Theorem:linear_lift}
Suppose $\tilde{\rho}$ is a projective representation of $\Gth$ formed by the fermionic MD $(\hS, \ \hT^2)$ of an SMC $\cB$, i.e. $\tilde{\rho}(s) = \hS, \ \tilde{\rho}(t^2) = \hT^2$. Then, $\tilde{\rho}$ always admits a  lift to a linear congruence representation of $\Gamma_\theta$.
\end{theorem}
\begin{proof}
By assumption, $\cB$ admits a minimal modular extension $\cC$. According to section 3.1 of ref. \cite{Cong_Subgroup_Super_Modular}, we can choose a particular basis so that $S$ and  $T^2$ of $\cC$ take the block-diagonal form
\begin{equation}
    S = \begin{pmatrix} \hS & 0 & 0 & 0 & 0 \\ 0 & 0& 2A & X & 0 \\ 0 & 2A^T & 0 & 0& 0 \\ 0 & 2X^T & 0 & 0 &0 \\ 0 &0 & 0 & 0& B  \end{pmatrix}, \quad T^2  = \begin{pmatrix} \hT^2 & 0 & 0 & 0 & 0 \\ 0 & \hT^2&  0& 0 & 0 \\ 0 & 0 & \hT_v^2 & 0& 0 \\ 0 & 0 & 0 & T^2_\sigma &0 \\ 0 &0 & 0 & 0& \hT_v^2  \end{pmatrix}
\end{equation}
where $\hS$ and $\hT^2$ are the modular matrices of $\cB$ (the other matrices such as $A$ and $X$ are not relevant for our purposes). This means that the projective $\SL$-representation formed by $S$ and $T$, after restriction to $\Gth$, becomes reducible (we are here interested only in representation-theoretic properties and are thus free to choose a basis).  More precisely, if we denote by $\tilde{\Phi}$ the projective $\SL$-representation formed by $S$ and $T$, we have
\begin{equation}
    \tilde{\Phi}|_{\Gth} = \tilde{\rho} \oplus \tilde{\Phi}'
\end{equation}
where $\tilde{\Phi}'$ is the remaining part. 

Let $N = \operatorname{ord} T.$ By Theorem II of ref. \cite{Cong_CFT}, $\tilde{\Phi}$ always admits a lift to a linear congruence representation $\Phi$ of level $n$ such that $N | n | 12 N$, which takes the form
\begin{align}
    &\Phi(s) = \tilde{\Phi}(s) \nonumber \\ 
    &\Phi(t) = e^{- 2 \pi i c/24} \tilde{\Phi}(t)
\end{align}
where $c$ is determined modulo $8$ by $\tilde{\Phi}$. (We can also always take the tensor product of this linear representation with $1$-dimensional representations of $\SL$, but this does not affect our argument.) The restriction of $\Phi$ to $\Gth$ then takes the form
\begin{align}
    &\Phi|_{\Gth}(s) = \hS \oplus ... \nonumber \\ 
    &\Phi|_{\Gth}(t^2) = e^{- 2 \pi i c/12} \hT^2 \oplus ... .
\end{align}
Since $\Phi|_{\Gth}$ is a linear representation of $\Gth$, its direct summand $\rho$ given by \begin{align}
\label{eq:linear_lift_c}
    &\rho(s) = \hS \nonumber \\ 
    &\rho(t) = e^{- 2 \pi i c/12} \hT^2
\end{align}
must also be a linear representation of $\Gth$.  Moreover, since the level of $\Phi$ is $n$, we note that 
\begin{equation}
    \ker \rho \geq \ker \Phi|_{\Gth} \geq \Gth \cap \Gamma(n). 
\end{equation}
Since $N$ is always even (since the list of simple objects of $\cC$ includes fermions, of spin $\frac{1}{2}$), $n$ is also even, so $\Gth \cap \Gamma(n) = \Gamma(n)$. Thus we have
\begin{equation}
    \ker \rho \geq \Gamma(n),
\end{equation}
i.e. the linear lift $\rho$ obtained by attaching a phase $e^{-2 \pi i c/12}$ to $\hT$, where $c$ is the central charge of one of the modular extensions, is congruence.

\end{proof}

Once we obtain the candidate $(\hS, \ \hT^2)$ via eq. \ref{eq:projectivize}, we check whether they are valid via the Verlinde formula (eq. \ref{eq:Verlinde}) and the Frobenius-Schur indicator condition (eq. \ref{eq:FS}). The results of this classification are summarized in Section \ref{Sec:Results} and appendix \ref{Sec:Table}.

\subsection{Central charge from linear representations}
\label{Sec:central_charge}

A strength of our approach is that we can determine the central charge of the resulting SMC, which is defined modulo $\frac{1}{2}$. For bosonic MTCs, whose central charge is defined modulo $8$, the approach of congruence representations confers no additional advantage as it is straightforward to determine the central charge from the modular matrices $S$ and $T$ via $(ST)^3 = e^{ 2\pi i c/ 8} S^2$. In the fermionic case, where $S$ is degenerate, and $\hS,\ \hT^2$ only form a projective representation of $\Gth$ rather than $\SL$, it is impossible to determine $c$ from the given MD by themselves using only the group relations of $\Gth$. Rather, for SMCs, $c$ is defined in terms of  the central charge of the modular extensions \cite{Wen2015Fermion}. While the central charge of each modular extension is defined modulo $8$, there are $16$ different modular extensions for a given SMC (as a consequence of Theorem 5.4 of ref. \cite{Lan_Kong_Wen_2016}) with their central charges differing by multiples of $\frac{1}{2}$ \cite{FermionicModular16}, so $c$ is defined modulo $\frac{1}{2}$ for an SMC. This means that, in order to compute the central charge of an SMC, we first need to compute (one of) the modular extensions. The modular extensions are bosonic MTCs of much higher rank (see Lemma 4.2 of \cite{ClassificationRank6} for an explicit bound on the rank), and their computation is a highly nontrivial task.

Our approach, which begins first with linear representations and then constructs the projective representations, allows us to determine the central charge of the SMCs we obtain without having to compute their modular extensions. The key idea is that, by  eq. \ref{eq:linear_lift_c}, the central charge of the modular extensions is involved in the lift of the fermionic MD to a linear representation. 

For each MD $(\hS, \ \hT^2)$, if $\rho(t) := e^{-2 \pi i \frac{c}{12}} \hT^2$ furnishes a linear lift, then $e^{ -2 \pi i \frac{c + m/2}{12}} \hT^2$ also furnishes a linear lift. Hence there are at least $24$ different linear representations (up to tensor product with $1$-dimensional representations, which does not affect the central charge) for a given MD. 


In our classification process, we start with a complete list of linear representations which can potentially yield valid MD. Thus, our list of linear representations must include these linear lifts coming from the existence of minimal modular extensions.  This means that, for every projective representation formed by  a given MD,   there are at least $24$ different linear representations which all lead to it. If there are exactly $24$, their $c$ should differ by multiples of $\frac{1}{2}$, and this fixes the $c$ of the SMC modulo $\frac{1}{2}$.

More concretely: consider a particular  pair $(\hS,\ \hT^2)$. We keep track of which linear representations $\rho_{\alpha}$ gave rise to this MD. These $\rho_{\alpha}$ differ from one another by a phase of $\rho_{\alpha}(t^2).$ If we find that there are 24 such $\rho_{\alpha}$ with $\rho_{\alpha}(t^2) = e^{2 \pi i \alpha/12} \rho_0(t^2)$, where $\rho_0$ is a chosen  reference representation, and $\alpha$ come in steps of $\frac{1}{2}$, then we can fix $c$ modulo $1/2$. For every MD we obtain, this has been the case, enabling us to determine $c$ modulo $\frac{1}{2}$.



\subsection{Unresolved representations}
\label{Sec:unres}
In Section \ref{Sec:resolved}, we saw that the possible orthogonal basis transformations $U_0$ are constrained to a finite set for resolved representations. For unresolved representations, we have a continuum of potential orthogonal transformations we need to apply and then check. 

In practice, however, all known cases are obtained from $\pi/4$ or $-\pi/4$ rotations, and we expect all valid MD will be obtained from orthogonal transformations involving simple angles, since the resulting $\hS$-matrix must solve the Verlinde formula condition that $\hS$ must diagonalize the fusion matrices, which are integer matrices. Hence, we check the following angles for unresolved Type (2), Type (2, 2), and Type (3) representations: $\pi/4, -\pi/4, \pi/6, -\pi/6, \pi/3, -\pi/3$. For Type (3), we allow the combination of all such rotations along all three axes of rotation. For Type (3, 2) and Type (4) representations, we only consider $\pi/4, -\pi/4$ rotations, but again along every possible axes of rotation. 

We find that, among dimension $5$ reducible representations, those  of $2+2+1$ Type (2, 2) and $2+1+1+1$ Type (4) yield valid MD. Among dimension $4$ reducible representations, those of $1+1+1+1$ Type (4) yield valid MD.   Every unitary MD obtained this way had previously been obtained in ref. \cite{Wen2015Fermion}, though we also obtain the non-unitary MD with the same fusion rules. In every case, the valid MD is obtained from a $\pi/4$ or $-\pi/4$ orthogonal transformation (for Type (4), we make two such orthogonal transformations along different axes). 

For unresolved Type (2) representations, there is only a single parameter $\phi$ for the orthogonal transformation, since we only have a rotation matrix on a $2$-dimensional subspace. In this case, we use Mathematica to directly solve for this unknown parameter given that all fusion coefficients (which depend on $\phi$) must be non-negative integers. We find that there is no solution, meaning we can definitively claim that  unresolved Type (2) representations do not yield valid MD. (Technically, because of numerical issues we need to specify some upper bound for the fusion coefficients, and in this case we only check up to $\hat{N}_{ij}^{ \ \ k} \leq 7$. However, we believe this should be sufficient, since the largest known fusion coefficient from valid MD is $4$.) 

Thus, our classification is complete for Type (2) unresolved representations (with the bound $\hat{N}_{ij}^{ \ \ k} \leq 7$); for other types, our classification may be incomplete and there may exist valid MD we have  missed. However, up to rank $10$, only a small minority of known MD come from unresolved representations, and in those few cases they are all obtained by orthogonal transformations involving only the angles $\pi/4$ or $-\pi/4$, so in practice it is unlikely that we have missed very many.

\section{Classification of fermionic modular data up to rank 10 \label{Sec:Results}}

Our results are tabulated in appendix \ref{Sec:Table}. Let us compare our results to previous results in the literature. 
\begin{itemize}
     \item Previous results (for any rank) were limited to  unitary MD, but we obtain both unitary and  non-unitary MD. We find that for every non-abelian fusion rule, there are both unitary and non-unitary MD realizing it (for comparison, in the  bosonic case, every non-abelian fusion rule up to rank $5$  has both a  unitary and non-unitary realization, but at rank $6$ there is a fusion rule which only admits a non-unitary realization \cite{rank6nonU, Reconstruction}). We expect that the non-unitary MD are related by Galois conjugation to the unitary MD. 
     
     \item In the unitary case, we recover all previously known MD \cite{Wen2015Fermion}. 
     
     \item In addition, we obtain two completely new fusion rules of rank $10$, and unitary and non-unitary MD realizing them. The unitary MD for these fusion rules have total quantum dimensions $D^2 = 472.379$ and $D^2 = 475.151$, which are much larger than any previously known total quantum dimension for rank $10$. Both of these are non-abelian and primitive, and do not fall into (the fermion condensation of) any known series of MTCs \cite{MTC_VOA}. These involve largest fusion coefficient $\hat{N}_{ij}^{ \ \ k} =  3$ and $4$, respectively. For comparison, all previously known examples of rank $10$ MD had the bound $\hat{N}_{ij}^{ \ \ k} \leq 2.$ Since our classification is in terms of the MD $(\hS, \ \hT^2)$ of the fermionic quotient, and hence the fusion rules $\hat{N}_{ij}^{ \ \ k}$ of the quotient, we need to check that full fusion rules $N_{ij}^{ \ \ k}$ exist which satisfy eqs. \ref{eq:fusion_quotient} and \ref{eq:balancing}. For these new fusion rules, we do find that such $N_{ij}^{ \ \ k}$ exist.

    \item The classification of rank $8$ fusion rules by ref. \cite{ClassificationRank8} had to place some bounds on the fusion coefficients in certain cases, though the bounds are very generous ($\hat{N}_{ij}^{ \ \ k} \leq 14$ or $\hat{N}_{ij}^{ \ \ k} \leq 21$). Our method places no such bound on the fusion coefficients, and yet do not find new fusion rules of rank $8$;  this is evidence for the results of ref. \cite{ClassificationRank8} being complete. 

    \item The arXiv version of ref. \cite{Wen2015Fermion} listed some MD which did not have valid modular extensions, colored in red. These do not appear in our classification as they do not form congruence representations. Our method automatically excludes such spurious MD without having to independently check the existence of modular extensions.

    \item As noted in section \ref{Sec:central_charge}, we find central charges modulo $\frac{1}{2}$ for every MD we obtain. Previously, the central charge data was missing for several MD of rank $8$ and rank $10$ in ref. \cite{Wen2015Fermion}.

\end{itemize}

We present the explicit MD of the two new classes  of rank-10 MD we have found. We show only one representative from each class. Here, $\chi^m_n = m + \sqrt{n}$.

\begin{equation}
\begin{aligned}
    \hS &= \frac{1}{\sqrt{30\chi^4_{15}}} \mqty(1 & \chi^4_{15} & \chi^5_{15} & \chi^3_{15} & \chi^3_{15} \\ \chi^4_{15} & 1 & \chi^5_{15} & -\chi^3_{15} & -\chi^3_{15} \\ \chi^5_{15} & \chi^5_{15} & -\chi^5_{15} & 0 & 0 \\ \chi^3_{15} & -\chi^3_{15} & 0 & \frac{1}{2}\chi^1_5\chi^3_{15} & -\frac{2\sqrt{30\chi^4_{15}}}{\chi^5_5} \\ \chi^3_{15} & -\chi^3_{15} & 0 & -\frac{2\sqrt{30\chi^4_{15}}}{\chi^5_5} & \frac{1}{2}\chi^1_5\chi^3_{15}), \\
    \hT^2 &= \mqty(\dmat{1,1,e^{2i\pi/3},e^{4i\pi/5},e^{-4i\pi/5}}), \quad D^2 = 472.379, \quad c = 0.
\end{aligned}
\end{equation}
\begin{equation}
\begin{aligned}
    \hS &= \frac{1}{2\sqrt{6\chi^5_{24}}} \mqty(1 & \chi^5_{24} & \chi^3_{6} & \chi^3_{6} & \chi^4_{24} \\ \chi^5_{24} & 1 & \chi^3_{6} & \chi^3_{6} & -\chi^4_{24} \\ \chi^3_{6} & \chi^3_{6} & -\chi^3_{6}-i\sqrt{6\chi^5_{24}} & -\chi^3_{6}+i\sqrt{6\chi^5_{24}} & 0 \\ \chi^3_{6} & \chi^3_{6} & -\chi^3_{6}+i\sqrt{6\chi^5_{24}} & -\chi^3_{6}-i\sqrt{6\chi^5_{24}} & 0 \\ \chi^4_{24} & -\chi^4_{24} & 0 & 0 & \chi^4_{24}), \\
    \hT^2 &= \mqty(\dmat{1,1,-1,-1,e^{2i\pi/3}}), \quad D^2 = 475.151, \quad c = 0.
\end{aligned}
\end{equation}

\section{Conclusion and outlook}

In this paper, we have detailed a procedure to classify the MD of SMCs using congruence representations, and have provided  a full classification of both unitary and non-unitary MD up to rank $10$. The classification is complete up to potential new MD coming from unresolved representations. Our result includes every unitary MD hitherto obtained, and also includes non-unitary MD. We also find new primitive MD of rank $10$ with completely new fusion rules.

The generalization to higher rank should be straightforward. There are several advantages to our approach compared to other approaches \cite{Wen2015Fermion, ClassificationRank6, ClassificationRank8}: (1) we do not need to place any bound on either the fusion coefficients or the total quantum dimension, (2) by treating unitary and non-unitary MD on an equal footing, we can easily obtain non-unitary as well as unitary MD, (3) we can determine the central charge without having to compute the modular extensions explicitly (4) spurious MD which do not admit a modular extension (arXiv version of ref. \cite{Wen2015Fermion}) are automatically excluded. Another advantage is that it allows to focus on non-split SMCs: in this paper we have included split SMCs as well as non-split SMCs to illustrate the power of this approach, but a classification of split SMCs are redundant since they follow trivially from the classification of MTCs. By excluding representations $\rho_{\rm isum}$ which are projectively extendable, we can automatically get rid of this redundancy and obtain a only the MD of non-split SMCs.  
A weakness of our approach is that it is difficult to handle unresolved representations.  In practice, a judicious choice of orthogonal transformations allows us to obtain some  valid MD even from unresolved representations, but we cannot claim the completeness of our classification. We leave a complete treatment of unresolved representations  to a future work.

Congruence representations also appear in rational conformal field theories (RCFTs): it is known that characters of RCFTs  transform as representations of $\SL$, and that these representations are congruence \cite{Cong_CFT}. Ref. \cite{Holomorphic_Modular} has used this idea to constrain the exponents of characters of bosonic RCFTs. Since the $\rm{NS}$-sector characters of a fermionic RCFT transforms as a representation of $\Gth$ \cite{bae2021fermionic}, congruence representations of $\Gth$ can readily be applied to fermionic RCFTs as well. Indeed, the NS-sector exponents of known fermionic characters \cite{bae2021fermionic, Bae_2022_fermionic_three} match the constraints obtained from our list of $\Gth$-congruence representations. A similar procedure can be applied to the R-sector, using a different subgroup of $\SL$. We will report our results on both the NS- and the R-sectors in an upcoming paper \cite{rcft_ours}.

\acknowledgments
We are grateful to Sungjay Lee for valuable discussions.
H.K., D.S. and G.Y.C. are supported by Samsung Science and Technology Foundation under Project Number SSTF-BA2002-05. H.K. is also supported by the National Research Foundation of Korea (NRF) grant funded by the Korea government (MSIT) (No. 2018R1D1A1B07042934).
M.Y. is supported under the YST program at the APCTP, which is funded through the Science and Technology Promotion Fund and Lottery Fund of the Korean Government, and also supported by
the Korean Local Governments of Gyeongsangbuk-do Province and Pohang City. G.Y.C is supported by the NRF of Korea (Grant No. 2020R1C1C1006048) funded by the Korean Government (MSIT) as well as the Institute of Basic Science under project code IBS-R014-D1. G.Y.C. is also supported by the Air Force Office of Scientific Research under Award No. FA2386-22-1-4061.

\bibliographystyle{JHEP}
\bibliography{ref}

\appendix

\section{Reverse induction formula \label{Sec:reverse_induction}}

Suppose we have a $3d$-dimensional symmetric $\SL$-rep given by $\qty(\cS, \cT)$. We assume that the spectrum of $T$ is non-degenerate. If the representation is an induced representation of some $d$-dimensional $\Gth$-rep given by $(\frS,\frT^2)$, then there exists a $3d \times 3d$ unitary matrix $\cU$ such that
\begin{equation} \label{eq:induced}
    \cU \cS \cU^{-1} = \mqty(\frS & 0 & 0 \\ 0 & 0 & \frS^2 \\ 0 & \mathds{1} & 0), \quad \cU \cT \cU^{-1} = \mqty(0 & \frT^2 & 0 \\ \mathds{1} & 0 & 0 \\ 0 & 0 & (\frS\frT^2)^{-1}).
\end{equation}
To find such $\cU$, first we re-arrange $T$ via a permutation matrix $P$ so that
\begin{equation} \label{eq:reverse_induction_permutation}
    P \cT P^{-1} = \mqty(\dmat{-\frT, \frT, \frT'}).
\end{equation}
Second, we introduce
\begin{equation} \label{eq:reverse_induction}
    U = \mqty(-\mathds{1} & \mathds{1} & 0 \\ \frT^{-1} & \frT^{-1} & 0 \\ 0 & 0 & C)
\end{equation}
where $C\frT'C^{-1} = (\frS\frT^2)^{-1}$, then eq.~(\ref{eq:induced}) is satisfied for $\cU = UP$. Note that the $U$ (\ref{eq:reverse_induction}) confines $\frS$ to a symmetric matrix. In addition, we have freedom of signed diagonal conjugation before conjugating with $U$. We denote the signed diagonal matrix by $D$. As a result, the transformation (\ref{eq:induced}) can be implemented by $\cU = UDP$. We can obtain a $d$-dimensional $\Gth$-rep given by $(\frS,\frT^2)$ from $3d$-dimensional $(\cS,\cT)$ via
\begin{equation} \label{eq:reverse_induction_long}
    \cU \cS \cU^{-1} = \mqty(\frS & 0 & 0 \\ 0 & 0 & \frS^2 \\ 0 & \mathds{1} & 0), \quad P \cT P^{-1} = \mqty(\dmat{-\frT, \frT, \frT'}).
\end{equation}

To efficiently implement above formula on a computer, we further simplify the procedure. Let 
\begin{equation}
    P\cS P^{-1} = \mqty(\cS_{11} & \cS_{12} & \cS_{13} \\ \cS_{21} & \cS_{22} & \cS_{23} \\ \cS_{31} & \cS_{32} & \cS_{33}), \quad D = \operatorname{diag}(a,b,c)
\end{equation}
where each $\cS_{ij}$ is a $d \times d$ matrix satisfying $\cS_{ji} = \cS_{ij}^T$ for all $i,j$, and $a,b,c$ are $d \times d$ signed diagonal matrices. Explicit calculation yields
\begin{equation} \label{eq:explicit_reverse_induction}
    \cU \cS \cU^{-1} = \frac{1}{2}\mqty(\cS_{11}^{aa}+\cS_{21}^{ba}+\cS_{12}^{ab}+\cS_{22}^{bb} & (-\cS_{11}^{aa}-\cS_{21}^{ba}+\cS_{12}^{ab}+\cS_{22}^{bb})\frT & * \\ -\frT^{-1}(\cS_{11}^{aa}-\cS_{21}^{ba}+\cS_{12}^{ab}-\cS_{22}^{bb}) & -\frT^{-1}(\cS_{11}^{aa}-\cS_{21}^{ba}-\cS_{12}^{ab}+\cS_{22}^{bb})\frT & * \\ * & * & *)
\end{equation}
where $\cS^{ab} = a\cS b^{-1}$ (same for similar notations) and irrelevant elements are denoted by $*$ for simplicity. Comparing eqs.~(\ref{eq:reverse_induction_long}) and eq.~(\ref{eq:explicit_reverse_induction}), we notice that
\begin{equation} \label{eq:reverse_induction_simple}
    \frS = 2\cS_{11}^{aa}, \quad \cS_{11}^{aa} = \cS_{12}^{ab} = \cS_{22}^{bb}.
\end{equation}
Therefore, for given $(\cS, \cT)$, we first find all permutation matrices $P$ satisfying eq.~(\ref{eq:reverse_induction_permutation}), and permute $(\cS, \cT)$ by them. Then, for each permutation, we check if eq.~(\ref{eq:reverse_induction_simple}) is satisfied for some signed diagonal matrices $a,b$. If eq.~(\ref{eq:reverse_induction_simple}) is satisfied, we store $(\frS,\frT^2)$.

\section{Congruence conditions \label{Sec:Cong}}

While the only relations among generators $s, \ t^2$ of  $\Gth$ are $s^4 =1$ and that $s^2$ commutes with $t^2,$ the generators satisfy many more relations in $\Gth/\Gamma(n)$, and any congruence representation of $\Gth$ needs its representation matrices to satisfy these relations. The precise number and content of these relations depend on the level $n.$ Since these relations are extra conditions which a representation $\rho$ needs to satisfy in order to be a congruence representation, we call them \emph{congruence conditions}. In this section, we list congruence conditions of $\Gth$-representations.

Our starting point is Theorem 1 of ref. \cite{Eholzer}, which lists the congruence conditions for representations of $\SL$. The expressions used there to write the conditions involve odd powers of $\rho(t) = \frT$, which are ill-defined for $\Gth$, and are unusable for our purposes. However, the conditions themselves are well-defined for $\Gth$; thus we need only re-write the expressions in terms of elements of $\Gth$. We have re-written expressions so that the conditions are written in terms of $\rho(s) = \frS$ and $\rho(t^2) = \frT^2$ only. Note that the level $n$ is always even for a $\Gth$-representation.

\begin{enumerate}
    \item $\frT^n = \mathds{1}$.

    \item For $a, b \in \mathbb{Z}_n^\times$, let $H(a) = \frS^2\frT^{a^2-a}\frS\frT^{-(\Bar{a}-1)}\frS(\hT^2\frS)^{a-1}$. 
    \begin{align}
        H(-1) &= \frS^2,\\
        H(a)H(b) &= H(ab),\\
        \frS H(a) &= H(\Bar{a})\frS,\\
        H(a) &= 
        \frS^2 \frT^{a^2-a} \frS \frT^{-(\bar{a} - 1)} \frS (\frT^2 \frS)^{a-1}  
    \end{align}
    where $\bar{a}$ is the multiplicative inverse of $a$ modulo $n.$

    \item $(\frS\frT^2)^n = \mathds{1}$.
    
    This condition is not independent of the above conditions,  but provides a simple check in many cases.  We can derive it by a direct calculation:
    \begin{equation}
        (st^2)^n = \mqty(1-n & -n \\ n & 1+n) \equiv \mqty(1 & 0 \\ 0 & 1) \mod{n}.
    \end{equation}

\end{enumerate}

    Projective congruence representations $\tilde{\rho}$ need only satisfy the above conditions up to an overall phase.

\section{Modular data from a direct sum of \texorpdfstring{$1$}{1}-dimensional representations}
\label{Sec:5overlaps}

In section \ref{Sec:resolved}, we have mentioned that a direct sum of five $1$-dimensional representations of $\Gth$ do not give rise to any valid MD. Let us call representations which are a direct sum of $1$-dimensional representations ``$1$d-sum representations,'' and SMCs arising from them as ``$1$d-sum SMCs.'' Here we prove that no $1$d-sum representations of dimension $5$ give rise to valid SMCs. We also prove a more general result on the type and  a $1$d-sum SMC 

Let $\rho = \bigoplus \limits_{i = 1}^{d}  \chi_i$ for some $1$-dimensional representations $\chi_i$. By the $t^2$-spectrum criterion (see section \ref{Sec:resolved}), $\rho(t^2)$ must be proportional to the identity. The resulting MD will be $\hT^2 = \mathds{1}$, and $T$ will consist of $1$s and $-1$s. Thus, the resulting SMC will have $\operatorname{ord} T = 2$. Spherical fusion categories (of which SMCs are a special case) satisfying this condition have been classified in ref. \cite{FSexp2Classification}. In particular, they find that any such spherical fusion category is pointed (abelian). On the other hand, any pointed SMC is split (see Proposition 2.1 of ref. \cite{ClassificationRank8}). Thus, there is no need for extra classification of $1$d-sum SMCs, as all of them come from stacking bosonic theories with $\cF_0$.

For the specific case of dimension $5$, we can simply look at the known bosonic classification in, say, ref. \cite{Reconstruction}, and see that there is no rank $5$ MTC whose $T$-matrix consists exclusively of $1$s and $-1$s. This proves the assertion in section \ref{Sec:resolved} that there are no rank $10$ $1$d-sum SMCs. On the other hand, in the bosonic classification of rank $4$ MTCs there is a well-known MTC whose $T$-matrix consists exclusively of $1$s and $-1$s: the toric code theory. Hence in rank 8 the toric code theory stacked with $\cF_0$ is a $1$d-sum SMC.

\section{Tables of modular data}
\label{Sec:Table}
In this section, we tabulate the  central charge, quantum dimensions, and topological spins of the obtained MD. The tables include data of both unitary and nonunitary MD, while those of nonunitary MD are shaded for legibility. MD with the same fusion rules are shown in a same box. They may be related by Galois conjugation. Table \ref{tab:rank4} contains rank $4$ MD, Table \ref{tab:rank6} rank 6, Tables \ref{tab:rank8a}, \ref{tab:rank8b} rank 8, and Tables \ref{tab:rank10}, \ref{tab:rank10b} rank 10.

For each rank, we enumerate the modular data, and subsequently refer to them in the format ${\rm Rank}^F_{\#}$. For example, the rank $4$ modular data with $D^2 = 13.656,$ which corresponds to the SMC ${\rm PSU}(2)_6$, shall be referred to as $4^F_{\#7}.$ 

In the comments, we specify whether a given MD is primitive. For primitive MD, if there is a known construction we specify it, e.g. ${\rm PSU}(2)_k = \cF_{(A_1, k)}$ (using the notation of ref. \cite{Wen2015Fermion}). If it is not primitive, we give the stacking structure. For bosonic MD referred to in the stacking structure, we use the enumeration in the Supplementary Materials of ref. \cite{Reconstruction} (Sections 5.3, 6.3, 7.3, and 8.3) as a reference system, and denote each modular data by ${\rm Rank}^B_{\#}$.  For example, the modular data of the Fibonacci MTC is referred to as $2^B_{\# 1}.$ Abelian modular data are not listed in ref. \cite{Wen2015Boson}; for these, we use the notation of ref. \cite{Wen2015Boson}, of the form $\rm{Rank}^B_{c \mod 8}$.  $\rm{Rank}^{B, *}_{c \mod 8}$ will refer to their non-unitary counterparts. When we stack two fermionic modular data, the stacking should take into account of the fact that both factors share the same local fermion $f$; this process is denoted $\boxtimes_{\cF_0}$, and we refer to refs. \cite{Wen2015Fermion, Lan_Kong_Wen_2016}
 for  details.

For the new classes of primitive modular data, $10^F_{\#59}$ to $10^F_{\# 70}$,  we specify the maximum fusion coefficient $\hat{N}$ via  $\hat{N} \leq 3$ or $\hat{N} \leq 4$.

\begin{table}[t]
    \centering
    \resizebox{\columnwidth}{!}{
    \begin{tabular}{|c|c|c|c|c|c|}
        \hline
        \# & $c$ & $D^2$ & Quantum dimensions & Topological spins & Comments \\
        \hline
        1 & 0 & 4 & $1,1,1,1$ & $0,\frac{1}{2},\frac{1}{4},-\frac{1}{4}$ & $\cF_0 \boxtimes 2_1^B $\\
        
        \rowcolor[gray]{0.9}
        2 & 0 & 4 & $1,1,-1,-1$ & $0,\frac{1}{2},\frac{1}{4},-\frac{1}{4}$ & $\cF_0 \boxtimes 2_1^{B, *}$  \\
        
        \hline
        3 & $\frac{1}{5}$ & $7.2360$ & $1,1,\zeta^1_3,\zeta^1_3$ & $0,\frac{1}{2},\frac{1}{10},-\frac{2}{5}$ & $\cF_0 \boxtimes 2^B_{\#2}$ \\
        
        4 & $-\frac{1}{5}$ & $7.2360$ & $1,1,\zeta^1_3,\zeta^1_3$ & $0,\frac{1}{2},-\frac{1}{10},\frac{2}{5}$ & $\cF_0 \boxtimes 2^B_{\#1}$ \\ 
        
        \rowcolor[gray]{0.9}5  & $\frac{1}{10}$ & 2.762 & $1,1,-\frac{1}{\zeta^1_3},-\frac{1}{\zeta^1_3}$ & $0,\frac{1}{2},-\frac{1}{5},\frac{3}{10}$ & $\cF_0 \boxtimes 2^B_{\# 4}$ \\
        
        \rowcolor[gray]{0.9} 6 & $-\frac{1}{10}$ &2.762 & $1,1,-\frac{1}{\zeta^1_3},-\frac{1}{\zeta^1_3}$ & $0,\frac{1}{2},\frac{1}{5},-\frac{3}{10}$ &  $\cF_0 \boxtimes 2^B_{\# 3}$\\
        
        \hline
        7 & $\frac{1}{4}$ & $13.656$ & $1,1,\chi^1_2,\chi^1_2$ & $0,\frac{1}{2},\frac{1}{4},-\frac{1}{4}$ & Primitive ($\cF_{(A_1,6)}$) \\
        
        \rowcolor[gray]{0.9} 8 & $\frac{1}{4}$ & 2.343 & $1,1,-\frac{1}{\chi^1_2},-\frac{1}{\chi^1_2}$ & $0,\frac{1}{2},\frac{1}{4},-\frac{1}{4}$ &  Primitive \\
        \hline
    \end{tabular}}
    \caption{List of rank 4 fermionic MD. Shaded data are of nonunitary MD. MD in the same box share the same (fermionic quotient) fusion rule. They may be related by Galois conjugation. For simplicity of notation, we introduce $\zeta^m_n = \frac{\sin[\pi(m+1)/(n+2)]}{\sin[\pi/(n+2)]}$ and $\chi^m_n = m + \sqrt{n}$. Data of bosonic MD are retrieved from ref. \cite{Reconstruction}.}
    \label{tab:rank4}
\end{table}

\begin{table}[t]
    \centering
    \resizebox{\columnwidth}{!}{
    \begin{tabular}{|c|c|c|c|c|c|}
        \hline
        \# & $c$ & $D^2$ & Quantum dimensions & Topological spins & Comments \\
        
        \hline
        1 & 0 & 6 & $1,1,1,1,1,1$ & $0,\frac{1}{2},\frac{1}{6},-\frac{1}{3}, \frac{1}{6}, -\frac{1}{3}$ & $\cF_0 \boxtimes 3^B_{2}$ \\
        
        2& 0 & 6 & $1,1,1,1,1,1$ & $0,\frac{1}{2},-\frac{1}{6},\frac{1}{3},-\frac{1}{6},\frac{1}{3}$ & $\cF_0 \boxtimes 3_{-2}^B $ \\
        \hline
        
        3& 0 & 8 & $1,1,1,1,\sqrt{2},\sqrt{2}$ & $0,\frac{1}{2},0,\frac{1}{2},\frac{1}{16},-\frac{7}{16}$ & $\cF_0 \boxtimes 3^B_{\# 7} = \cF_0 \boxtimes 3^B_{\# 9}$ \\
        
        4& 0 & 8 & $1,1,1,1,\sqrt{2},\sqrt{2}$ & $0,\frac{1}{2},0,\frac{1}{2},-\frac{1}{16},\frac{7}{16}$ & $\cF_0 \boxtimes 3^B_{\# 8} = \cF_0 \boxtimes 3^B_{\# 10}$
        \\
        
        5 & 0 & 8 & $1,1,1,1,\sqrt{2},\sqrt{2}$ & $0,\frac{1}{2},0,\frac{1}{2},\frac{3}{16},-\frac{5}{16}$ & $\cF_0 \boxtimes 3^B_{\# 15} = \cF_0 \boxtimes 3^B_{\# 17}$ \\
        
        6 & 0 & 8 & $1,1,1,1,\sqrt{2},\sqrt{2}$ & $0,\frac{1}{2},0,\frac{1}{2},-\frac{3}{16},\frac{5}{16}$ & $\cF_0 \boxtimes 3^B_{\# 16} = \cF_0 \boxtimes 3^B_{\# 18}$ \\
        
        \rowcolor[gray]{0.9}
        7 & 0 & 8 & $1,1,1,1,-\sqrt{2},-\sqrt{2}$ & $0,\frac{1}{2},0,\frac{1}{2},-\frac{3}{16},\frac{5}{16}$ & $\cF_0 \boxtimes 3^B_{\# 12} = \cF_0 \boxtimes 3^B_{\# 14}$ \\
        
        \rowcolor[gray]{0.9}
        8 & 0 & 8 & $1,1,1,1,-\sqrt{2},-\sqrt{2}$ & $0,\frac{1}{2},0,\frac{1}{2},-\frac{1}{16},\frac{7}{16}$ &  $\cF_0 \boxtimes 3^B_{\# 20} =\cF_0 \boxtimes 3^B_{\# 22}$\\
        
        \rowcolor[gray]{0.9}
        9 & 0 & 8 & $1,1,1,1,-\sqrt{2},-\sqrt{2}$ & $0,\frac{1}{2},0,\frac{1}{2},\frac{1}{16},-\frac{7}{16}$ &  $\cF_0 \boxtimes 3^B_{\# 19} =\cF_0 \boxtimes 3^B_{\# 21}$\\
        
        \rowcolor[gray]{0.9}
        10 & 0 & 8 & $1,1,1,1,-\sqrt{2},-\sqrt{2}$ & $0,\frac{1}{2},0,\frac{1}{2},\frac{3}{16},-\frac{5}{16}$ &  $\cF_0 \boxtimes 3^B_{\# 11} =\cF_0 \boxtimes 3^B_{\# 13}$\\
        
        \hline
        11 & $\frac{1}{7}$ & 18.591 & $1,1,\zeta^1_5,\zeta^1_5,\zeta^2_5,\zeta^2_5$ & $0,\frac{1}{2},-\frac{1}{7},\frac{5}{14},-\frac{3}{14},\frac{2}{7}$ & $\cF_0 \boxtimes 3^B_{\# 2}$ \\
        
        12 & $-\frac{1}{7}$ & 18.591 & $1,1,\zeta^1_5,\zeta^1_5,\zeta^2_5,\zeta^2_5$ & $0,\frac{1}{2},\frac{1}{7},-\frac{5}{14},\frac{3}{14},-\frac{2}{7}$ & $\cF_0 \boxtimes 3^B_{\# 1}$ \\
        
        \rowcolor[gray]{0.9}
        13 & $-\frac{3}{14}$ & 5.724& $1,1,-\frac{\zeta^2_5}{\zeta^1_5},-\frac{\zeta^2_5}{\zeta^1_5},\frac{1}{\zeta^1_5},\frac{1}{\zeta^1_5}$ & $0,\frac{1}{2},\frac{3}{14},-\frac{2}{7},\frac{1}{14},-\frac{3}{7}$ & $\cF_0 \boxtimes 3^B_{\# 5}$ \\
        
        \rowcolor[gray]{0.9}
        14 & $\frac{1}{14}$ & 3.682 & $1,1,-\frac{\zeta^1_5}{\zeta^2_5},-\frac{\zeta^1_5}{\zeta^2_5},\frac{1}{\zeta^2_5},\frac{1}{\zeta^2_5}$ & $0,\frac{1}{2},\frac{1}{7},-\frac{5}{14},-\frac{1}{14},\frac{3}{7}$ & $\cF_0 \boxtimes 3^B_{\# 3}$ \\
        
        \rowcolor[gray]{0.9}
        15 & $-\frac{1}{14}$ & 3.682 & $1,1,-\frac{\zeta^1_5}{\zeta^2_5},-\frac{\zeta^1_5}{\zeta^2_5},\frac{1}{\zeta^2_5},\frac{1}{\zeta^2_5}$ & $0,\frac{1}{2},-\frac{1}{7},\frac{5}{14},\frac{1}{14},-\frac{3}{7}$ &$\cF_0 \boxtimes 3^B_{\# 6}$  \\
        
        \rowcolor[gray]{0.9}
        16 & $\frac{3}{14}$ & 5.724& $1,1,-\frac{\zeta^2_5}{\zeta^1_5},-\frac{\zeta^2_5}{\zeta^1_5},\frac{1}{\zeta^1_5},\frac{1}{\zeta^1_5}$ & $0,\frac{1}{2},-\frac{3}{14},\frac{2}{7},-\frac{1}{14},\frac{3}{7}$ &$\cF_0 \boxtimes 3^B_{\# 4}$  \\
        
        \hline
        17 & 0 & 44.784 & $1,1,\chi^1_3,\chi^1_3,\chi^2_3,\chi^2_3$ & $0,\frac{1}{2},-\frac{1}{6},\frac{1}{3},0,\frac{1}{2}$ & Primitive ($\cF_{(A_1,-10)}$ )\\
        
        18 & 0 & 44.784 & $1,1,\chi^1_3,\chi^1_3,\chi^2_3,\chi^2_3$ & $0,\frac{1}{2},\frac{1}{6},-\frac{1}{3},0,\frac{1}{2}$ &  Primitive ($\cF_{(A_1,10)}$ ) \\
        
        \rowcolor[gray]{0.9}
        19 & 0 & 3.2154 & $1,1,-\frac{\chi^1_3}{\chi^2_3},-\frac{\chi^1_3}{\chi^2_3},\frac{1}{\chi^2_3},\frac{1}{\chi^2_3}$ & $0,\frac{1}{2},-\frac{1}{6},\frac{1}{3},0,\frac{1}{2}$ &  Primitive \\
        
        \rowcolor[gray]{0.9}
        20 & 0 & 3.2154 & $1,1,-\frac{\chi^1_3}{\chi^2_3},-\frac{\chi^1_3}{\chi^2_3},\frac{1}{\chi^2_3},\frac{1}{\chi^2_3}$ & $0,\frac{1}{2},\frac{1}{6},-\frac{1}{3},0,\frac{1}{2}$ &  Primitive \\
        \hline
    \end{tabular}}
    \caption{List of rank 6 fermionic MD.}
    \label{tab:rank6}
\end{table}

\begin{table}[t]
    \centering
    \resizebox{\columnwidth}{!}{
    \begin{tabular}{|c|c|c|c|c|c|}
        \hline
        \# & $c$ & $D^2$ & Quantum dimensions & Topological spins & Comments \\
        \hline
        1 & 0 & 8 & $1,1,1,1,1,1,1,1$ & $0,\frac{1}{2},0,\frac{1}{2},0,\frac{1}{2},0,\frac{1}{2}$ & $\cF_0 \boxtimes 4_{0}^{B,a} $ \\
        
        \rowcolor[gray]{0.9}
        2 & 0 & 8 & $1,1,1,1,-1,-1,-1,-1$ & $0,\frac{1}{2},0,\frac{1}{2},0,\frac{1}{2},0,\frac{1}{2}$ & $\cF_0 \boxtimes 4_{0}^{B,a,*} $\\
        
        \hline
        3& 0 & 8 & $1,1,1,1,1,1,1,1$ & $0,\frac{1}{2},0,\frac{1}{2},\frac{1}{8},-\frac{3}{8},\frac{1}{8},-\frac{3}{8}$ & $\cF_0 \boxtimes 4_{1}^B$ \\
        
        4& 0 & 8 & $1,1,1,1,1,1,1,1$ & $0,\frac{1}{2},0,\frac{1}{2},-\frac{1}{8},\frac{3}{8},-\frac{1}{8},\frac{3}{8}$ & $\cF_0 \boxtimes 4_{3}^B$ \\
        
        \rowcolor[gray]{0.9}
        5 & 0 & 8 & $1,1,-1,-1,-1,-1,1,1$ & $0,\frac{1}{2},0,\frac{1}{2},-\frac{1}{8},\frac{3}{8},-\frac{1}{8},\frac{3}{8}$ & $\cF_0 \boxtimes 4_{3}^{B,*}$\\
        
        \rowcolor[gray]{0.9}
        6 & 0 & 8 & $1,1,-1,-1,-1,-1,1,1$ & $0,\frac{1}{2},0,\frac{1}{2},\frac{1}{8},-\frac{3}{8},\frac{1}{8},-\frac{3}{8}$ & $\cF_0 \boxtimes 4_{1}^{B,*}$\\
        
        \hline
        7 & 0 & 8 & $1,1,1,1,1,1,1,1$ & $0,\frac{1}{2},0,\frac{1}{2},\frac{1}{4},-\frac{1}{4},\frac{1}{4},-\frac{1}{4}$ & $\cF_0 \boxtimes 4_{0}^{B,b}$ \\
        
        \rowcolor[gray]{0.9}
        8 & 0 & 8 & $1,1,-1,-1,-1,-1,1,1$ & $0,\frac{1}{2},0,\frac{1}{2},\frac{1}{4},-\frac{1}{4},\frac{1}{4},-\frac{1}{4}$ & $\cF_0 \boxtimes 4_{0}^{B,b, *1}$\\
        
        \rowcolor[gray]{0.9}
        9 & 0 & 8 & $1,1,1,1,-1,-1,-1,-1$ & $0,\frac{1}{2},0,\frac{1}{2},\frac{1}{4},-\frac{1}{4},\frac{1}{4},-\frac{1}{4}$ &  $\cF_0 \boxtimes 4_{0}^{B,b, *2}$\\
        
        \hline
        10 & $\frac{1}{5}$ & 14.472 & $1,1,1,1,\zeta^1_3,\zeta^1_3,\zeta^1_3,\zeta^1_3$ & $0,\frac{1}{2},\frac{1}{4},-\frac{1}{4},\frac{1}{10},-\frac{2}{5},-\frac{3}{20},\frac{7}{20}$ & $\cF_0 \boxtimes 4^B_{\# 18} = \cF_0\boxtimes 4^B_{\# 20}$ \\
        
        11 & $-\frac{1}{5}$ & 14.472 & $1,1,1,1,\zeta^1_3,\zeta^1_3,\zeta^1_3,\zeta^1_3$ & $0,\frac{1}{2},\frac{1}{4},-\frac{1}{4},\frac{3}{20},-\frac{7}{20},-\frac{1}{10},\frac{2}{5}$ & $\cF_0 \boxtimes 4^B_{\# 17} = \cF_0\boxtimes 4^B_{\# 19}$ \\
        
        \rowcolor[gray]{0.9}
        12 & $-\frac{1}{5}$ & 14.472 & $1,1,-1,-1,-\zeta^1_3,-\zeta^1_3,\zeta^1_3,\zeta^1_3$ & $0,\frac{1}{2},\frac{1}{4},-\frac{1}{4},\frac{3}{20},-\frac{7}{20},-\frac{1}{10},\frac{2}{5}$ & $\cF_0 \boxtimes 4^B_{\# 26} =\cF_0 \boxtimes 4^B_{\# 27}$ \\
        
        \rowcolor[gray]{0.9}
        13 & $\frac{1}{10}$ & 5.528 & $1,1,-1,-1,-\frac{1}{\zeta^1_3},-\frac{1}{\zeta^1_3},\frac{1}{\zeta^1_3},\frac{1}{\zeta^1_3}$ & $0,\frac{1}{2},\frac{1}{4},-\frac{1}{4},-\frac{1}{5},\frac{3}{10},\frac{1}{20},-\frac{9}{20}$ & $\cF_0 \boxtimes 4^B_{\# 25} = \cF_0 \boxtimes 4^B_{\# 29}$\\
        
        \rowcolor[gray]{0.9}
        14 & $\frac{1}{10}$ & 5.528 & $1,1,1,1,-\frac{1}{\zeta^1_3},-\frac{1}{\zeta^1_3},-\frac{1}{\zeta^1_3},-\frac{1}{\zeta^1_3}$ & $0,\frac{1}{2},\frac{1}{4},-\frac{1}{4},\frac{1}{20},-\frac{9}{20},-\frac{1}{5},\frac{3}{10}$ & $\cF_0 \boxtimes 4^B_{\# 22} = \cF_0 \boxtimes 4^B_{\# 24}$\\
        
        \rowcolor[gray]{0.9}
        15 & $-\frac{1}{10}$ & 5.528 & $1,1,-1,-1,-\frac{1}{\zeta^1_3},-\frac{1}{\zeta^1_3},\frac{1}{\zeta^1_3},\frac{1}{\zeta^1_3}$ & $0,\frac{1}{2},\frac{1}{4},-\frac{1}{4},\frac{1}{5},-\frac{3}{10},-\frac{1}{20},\frac{9}{20}$ & $\cF_0 \boxtimes 4^B_{\# 28} = \cF_0 \boxtimes 4^B_{\# 32}$\\
        
        \rowcolor[gray]{0.9} 
        16 & $-\frac{1}{10}$ & 5.528 & $1,1,1,1,-\frac{1}{\zeta^1_3},-\frac{1}{\zeta^1_3},-\frac{1}{\zeta^1_3},-\frac{1}{\zeta^1_3}$ & $0,\frac{1}{2},\frac{1}{4},-\frac{1}{4},\frac{1}{5},-\frac{3}{10},-\frac{1}{20},\frac{9}{20}$ &$\cF_0 \boxtimes 4^B_{\# 21} = \cF_0 \boxtimes 4^B_{\# 23}$ \\
        
        \rowcolor[gray]{0.9}
        17 & $\frac{1}{5}$ & 14.472& $1,1,-1,-1,-\zeta^1_3,-\zeta^1_3,\zeta^1_3,\zeta^1_3$ & $0,\frac{1}{2},\frac{1}{4},-\frac{1}{4},-\frac{3}{20},\frac{7}{20},\frac{1}{10},-\frac{2}{5}$ &$\cF_0 \boxtimes 4^B_{\# 30} = \cF_0 \boxtimes 4^B_{\# 31}$ \\
        \hline
        
        18 & 0 & 24 & $1,1,1,1,2,2,\sqrt{6},\sqrt{6}$ & $0,\frac{1}{2},0,\frac{1}{2},\frac{1}{6},-\frac{1}{3},\frac{1}{16},-\frac{7}{16}$ & Primitive ($\cF_{U(1)_6/\mathbb{Z}_2}$) \\
        
        19 & 0 & 24 & $1,1,1,1,2,2,\sqrt{6},\sqrt{6}$ & $0,\frac{1}{2},0,\frac{1}{2},\frac{1}{6},-\frac{1}{3},-\frac{1}{16},\frac{7}{16}$ & Primitive \\
        
        20 & 0 & 24 & $1,1,1,1,2,2,\sqrt{6},\sqrt{6}$ & $0,\frac{1}{2},0,\frac{1}{2},\frac{1}{6},-\frac{1}{3},\frac{3}{16},-\frac{5}{16}$ & Primitive \\
        
        21 & 0 & 24 & $1,1,1,1,2,2,\sqrt{6},\sqrt{6}$ & $0,\frac{1}{2},0,\frac{1}{2},\frac{1}{6},-\frac{1}{3},-\frac{3}{16},\frac{5}{16}$ & Primitive \\
        
        22 & 0 & 24 & $1,1,1,1,2,2,\sqrt{6},\sqrt{6}$ & $0,\frac{1}{2},0,\frac{1}{2},-\frac{1}{6},\frac{1}{3},\frac{1}{16},-\frac{7}{16}$ & Primitive \\
        
        23 & 0 & 24 & $1,1,1,1,2,2,\sqrt{6},\sqrt{6}$ & $0,\frac{1}{2},0,\frac{1}{2},-\frac{1}{6},\frac{1}{3},-\frac{1}{16},\frac{7}{16}$ & Primitive \\
        
        24 & 0 & 24 & $1,1,1,1,2,2,\sqrt{6},\sqrt{6}$ & $0,\frac{1}{2},0,\frac{1}{2},-\frac{1}{6},\frac{1}{3},\frac{3}{16},-\frac{5}{16}$ & Primitive \\
        
        25 & 0 & 24 & $1,1,1,1,2,2,\sqrt{6},\sqrt{6}$ & $0,\frac{1}{2},0,\frac{1}{2},-\frac{1}{6},\frac{1}{3},-\frac{3}{16},\frac{5}{16}$ & Primitive \\
        
        \rowcolor[gray]{0.9}
        26 & 0 & 24 & $1,1,1,1,2,2,-\sqrt{6},-\sqrt{6}$ & $0,\frac{1}{2},0,\frac{1}{2},-\frac{1}{6},\frac{1}{3},-\frac{1}{16},\frac{7}{16}$ & Primitive\\
        
        \rowcolor[gray]{0.9}
        27 & 0 & 24 & $1,1,1,1,2,2,-\sqrt{6},-\sqrt{6}$ & $0,\frac{1}{2},0,\frac{1}{2},\frac{1}{6},-\frac{1}{3},-\frac{1}{16},\frac{7}{16}$ &Primitive \\
        
        \rowcolor[gray]{0.9}
        28 & 0 & 24 & $1,1,1,1,2,2,-\sqrt{6},-\sqrt{6}$ & $0,\frac{1}{2},0,\frac{1}{2},\frac{1}{6},-\frac{1}{3},-\frac{3}{16},\frac{5}{16}$ & Primitive\\
        
        \rowcolor[gray]{0.9}
        29 & 0 & 24 & $1,1,1,1,2,2,-\sqrt{6},-\sqrt{6}$ & $0,\frac{1}{2},0,\frac{1}{2},\frac{1}{6},-\frac{1}{3},\frac{3}{16},-\frac{5}{16}$ &Primitive \\
        
        \rowcolor[gray]{0.9}
        30 & 0 & 24 & $1,1,1,1,2,2,-\sqrt{6},-\sqrt{6}$ & $0,\frac{1}{2},0,\frac{1}{2},-\frac{1}{6},\frac{1}{3},\frac{1}{16},-\frac{7}{16}$ & Primitive\\
        
        \rowcolor[gray]{0.9}
        31 & 0 & 24 & $1,1,1,1,2,2,-\sqrt{6},-\sqrt{6}$ & $0,\frac{1}{2},0,\frac{1}{2},\frac{1}{6},-\frac{1}{3},\frac{1}{16},-\frac{7}{16}$ &Primitive \\
        
        \rowcolor[gray]{0.9}
        32 & 0 & 24 & $1,1,1,1,2,2,-\sqrt{6},-\sqrt{6}$ & $0,\frac{1}{2},0,\frac{1}{2},-\frac{1}{6},\frac{1}{3},-\frac{3}{16},\frac{5}{16}$ & Primitive\\
        
        \rowcolor[gray]{0.9}
        33 & 0 & 24 & $1,1,1,1,2,2,-\sqrt{6},-\sqrt{6}$ & $0,\frac{1}{2},0,\frac{1}{2},-\frac{1}{6},\frac{1}{3},\frac{3}{16},-\frac{5}{16}$ & Primitive\\
        \hline
    \end{tabular}}
    \caption{List of rank 8 fermionic MD.}
    \label{tab:rank8a}
\end{table}

\begin{table}[t]
    \centering
    \resizebox{\columnwidth}{!}{
    \begin{tabular}{|c|c|c|c|c|c|}
        \hline
        \# & $c$ & $D^2$ & Quantum dimensions & Topological spins & Comments \\
        \hline
        
        34 & $-\frac{1}{10}$ & 26.180 & $1,1,\zeta^1_3,\zeta^1_3,\zeta^1_3,\zeta^1_3,\zeta^2_8,\zeta^2_8$ & $0,\frac{1}{2},\frac{1}{10},-\frac{2}{5},\frac{1}{10},-\frac{2}{5},\frac{1}{5},-\frac{3}{10}$ & $\cF_0 \boxtimes 4^B_{\# 2}$ \\
        
        35 & 0 & 26.180 & $1,1,\zeta^1_3,\zeta^1_3,\zeta^1_3,\zeta^1_3,\zeta^2_8,\zeta^2_8$ & $0,\frac{1}{2},\frac{1}{10},-\frac{2}{5},-\frac{1}{10},\frac{2}{5},0,\frac{1}{2}$ & $\cF_0 \boxtimes 4_{\# 5}^B$ \\
        
        36 & $\frac{1}{10}$ & 26.180 & $1,1,\zeta^1_3,\zeta^1_3,\zeta^1_3,\zeta^1_3,\zeta^2_8,\zeta^2_8$ & $0,\frac{1}{2},-\frac{1}{10},\frac{2}{5},-\frac{1}{10},\frac{2}{5},-\frac{1}{5},\frac{3}{10}$ & $\cF_0 \boxtimes 4_{\# 1}^B$ 
        \\
        
        \rowcolor[gray]{0.9}
        37 & 0 & 3.820 & $1,1,-\frac{1}{\zeta^1_3},-\frac{1}{\zeta^1_3},-\frac{1}{\zeta^1_3},-\frac{1}{\zeta^1_3},\frac{1}{\zeta^2_8},\frac{1}{\zeta^2_8}$ & $0,\frac{1}{2},\frac{1}{5},-\frac{3}{10},-\frac{1}{5},\frac{3}{10},0,\frac{1}{2}$ & $\cF_0 \boxtimes 4^B_{\# 6}$
        \\
        
        \rowcolor[gray]{0.9}
        38 & $-\frac{1}{5}$ & 3.820 & $1,1,-\frac{1}{\zeta^1_3},-\frac{1}{\zeta^1_3},-\frac{1}{\zeta^1_3},-\frac{1}{\zeta^1_3},\frac{1}{\zeta^2_8},\frac{1}{\zeta^2_8}$ & $0,\frac{1}{2},\frac{1}{5},-\frac{3}{10},\frac{1}{5},-\frac{3}{10},-\frac{1}{10},\frac{2}{5}$ & $\cF_0 \boxtimes 4^B_{\# 3}$
        \\
        
        \rowcolor[gray]{0.9}
        39 & $\frac{1}{5}$ & 3.820 & $1,1,-\frac{1}{\zeta^1_3},-\frac{1}{\zeta^1_3},-\frac{1}{\zeta^1_3},-\frac{1}{\zeta^1_3},\frac{1}{\zeta^2_8},\frac{1}{\zeta^2_8}$ & $0,\frac{1}{2},-\frac{1}{5},\frac{3}{10},-\frac{1}{5},\frac{3}{10},\frac{1}{10},-\frac{2}{5}$ & $\cF_0 \boxtimes 4^B_{\# 4}$
        \\
        
        \rowcolor[gray]{0.9}
        40 & $-\frac{1}{5}$ & 10 & $1,1,-1,-1,-\frac{1}{\zeta^1_3},-\frac{1}{\zeta^1_3},\zeta^1_3,\zeta^1_3$ & $0,\frac{1}{2},-\frac{1}{10},\frac{2}{5},-\frac{1}{5},\frac{3}{10},\frac{1}{10},-\frac{2}{5}$ & $\cF_0 \boxtimes 4^B_{\# 10}$
        \\
        
        \rowcolor[gray]{0.9}
        41 & $-\frac{1}{10}$ & 10 & $1,1,-1,-1,-\frac{1}{\zeta^1_3},-\frac{1}{\zeta^1_3},\zeta^1_3,\zeta^1_3$ & $0,\frac{1}{2},-\frac{1}{5},\frac{3}{10},\frac{1}{5},-\frac{3}{10},\frac{1}{10},-\frac{2}{5}$ & $\cF_0 \boxtimes 4^B_{\# 9}$
        \\
        
        \rowcolor[gray]{0.9}
        42 & $\frac{1}{10}$ & 10& $1,1,-1,-1,-\frac{1}{\zeta^1_3},-\frac{1}{\zeta^1_3},\zeta^1_3,\zeta^1_3$ & $0,\frac{1}{2},\frac{1}{5},-\frac{3}{10},-\frac{1}{5},\frac{3}{10},-\frac{1}{10},\frac{2}{5}$ & $\cF_0 \boxtimes 4^B_{\# 8}$
        \\
        
        \rowcolor[gray]{0.9}
        43 & $\frac{1}{5}$ & 10& $1,1,-1,-1,-\frac{1}{\zeta^1_3},-\frac{1}{\zeta^1_3},\zeta^1_3,\zeta^1_3$ & $0,\frac{1}{2},\frac{1}{10},-\frac{2}{5},\frac{1}{5},-\frac{3}{10},-\frac{1}{10},\frac{2}{5}$ & $\cF_0 \boxtimes 4^B_{\# 7}$
        \\
        \hline
        
        44 & $\frac{1}{4}$ & 27.313 & $1,1,1,1,\chi^1_2,\chi^1_2,\chi^1_2,\chi^1_2$ & $0,\frac{1}{2},\frac{1}{4},-\frac{1}{4},0,\frac{1}{2},\frac{1}{4},-\frac{1}{4}$ & $4^F_{\# 7} \boxtimes 2_{1}^B$ 
        \\
        
        \rowcolor[gray]{0.9}
        45 & $\frac{1}{4}$ & 27.313 & $1,1,-1,-1,-\chi^1_2,-\chi^1_2,\chi^1_2,\chi^1_2$ & $0,\frac{1}{2},\frac{1}{4},-\frac{1}{4},0,\frac{1}{2},\frac{1}{4},-\frac{1}{4}$ & $4^F_{\# 7} \boxtimes 2_{1}^{B,*}$  
        \\
        
        \rowcolor[gray]{0.9}
        46 & $\frac{1}{4}$ & 4.6863 & $1,1,-1,-1,-\frac{1}{\chi^1_2},-\frac{1}{\chi^1_2},\frac{1}{\chi^1_2},\frac{1}{\chi^1_2}$ & $0,\frac{1}{2},\frac{1}{4},-\frac{1}{4},\frac{1}{4},-\frac{1}{4},0,\frac{1}{2}$ &  $4^F_{\# 8} \boxtimes 2_{1}^B$ 
        \\
        
        \rowcolor[gray]{0.9}
        47 & $\frac{1}{4}$ & 4.6863 & $1,1,1,1,-\frac{1}{\chi^1_2},-\frac{1}{\chi^1_2},-\frac{1}{\chi^1_2},-\frac{1}{\chi^1_2}$ & $0,\frac{1}{2},\frac{1}{4},-\frac{1}{4},0,\frac{1}{2},\frac{1}{4},-\frac{1}{4}$ &  $4^F_{\# 8} \boxtimes 2_{1}^{B,*}$
        \\
        
        \hline
        
        48 & $\frac{1}{6}$ & 38.468 & $1,1,\zeta^1_7,\zeta^1_7,\zeta^2_7,\zeta^2_7,\zeta^3_7,\zeta^3_7$ & $0,\frac{1}{2},\frac{1}{6},-\frac{1}{3},-\frac{2}{9},\frac{5}{18},-\frac{1}{6},\frac{1}{3}$ & $\cF_0 \boxtimes 4_{\# 12}^B$ 
        \\
        
        49 & $-\frac{1}{6}$ & 38.468 & $1,1,\zeta^1_7,\zeta^1_7,\zeta^2_7,\zeta^2_7,\zeta^3_7,\zeta^3_7$ & $0,\frac{1}{2},-\frac{1}{6},\frac{1}{3},\frac{2}{9},-\frac{5}{18},\frac{1}{6},-\frac{1}{3}$ & $\cF_0 \boxtimes 4_{\# 11}^B$ 
        \\
        
        \rowcolor[gray]{0.9}
        50 & $-\frac{1}{6}$ & 10.890 & $1,1,-\frac{\zeta^3_7}{\zeta^1_7},-\frac{\zeta^3_7}{\zeta^1_7},-\frac{\zeta^3_7}{\zeta^1_7},-\frac{\zeta^3_7}{\zeta^1_7},\frac{\zeta^2_7}{\zeta^1_7},\frac{\zeta^2_7}{\zeta^1_7}$ & $0,\frac{1}{2},-\frac{1}{6},\frac{1}{3},\frac{1}{6},-\frac{1}{3},\frac{1}{18},-\frac{4}{9}$ & $\cF_0 \boxtimes 4^B_{\# 15}$ 
        \\
        
        \rowcolor[gray]{0.9}
        51 & $-\frac{1}{6}$ & 4.640 & $1,1,-\frac{\zeta^2_7}{\zeta^3_7},-\frac{\zeta^2_7}{\zeta^3_7},-\frac{1}{\zeta^3_7},-\frac{1}{\zeta^3_7},\frac{\zeta^1_7}{\zeta^3_7},\frac{\zeta^1_7}{\zeta^3_7}$ & $0,\frac{1}{2},-\frac{1}{9},\frac{7}{18},-\frac{1}{6},\frac{1}{3},\frac{1}{6},-\frac{1}{3}$ &  $\cF_0 \boxtimes 4^B_{\# 16}$
        \\
        
        \rowcolor[gray]{0.9}
        52 & $\frac{1}{6}$ & 10.890 & $1,1,-\frac{\zeta^3_7}{\zeta^1_7},-\frac{\zeta^3_7}{\zeta^1_7},-\frac{\zeta^3_7}{\zeta^1_7},-\frac{\zeta^3_7}{\zeta^1_7},\frac{\zeta^2_7}{\zeta^1_7},\frac{\zeta^2_7}{\zeta^1_7}$ & $0,\frac{1}{2},\frac{1}{6},-\frac{1}{3},-\frac{1}{6},\frac{1}{3},-\frac{1}{18},\frac{4}{9}$ &  $\cF_0 \boxtimes 4^B_{\# 14}$
        \\
        
        \rowcolor[gray]{0.9}
        53 & $\frac{1}{6}$ & 4.640 & $1,1,-\frac{\zeta^2_7}{\zeta^3_7},-\frac{\zeta^2_7}{\zeta^3_7},-\frac{1}{\zeta^3_7},-\frac{1}{\zeta^3_7},\frac{\zeta^1_7}{\zeta^3_7},\frac{\zeta^1_7}{\zeta^3_7}$ & $0,\frac{1}{2},\frac{1}{9},-\frac{7}{18},\frac{1}{6},-\frac{1}{3},-\frac{1}{6},\frac{1}{3}$ &  $\cF_0 \boxtimes 4^B_{\# 13}$
        \\
        \hline
        
        54 & $-\frac{1}{20}$ & 49.410 & $1,1,\zeta^1_3\zeta^2_6,\zeta^1_3\zeta^2_6,\zeta^2_6,\zeta^2_6,\zeta^1_3,\zeta^1_3$ & $0,\frac{1}{2},-\frac{3}{20},\frac{7}{20},\frac{1}{4},-\frac{1}{4},\frac{1}{10},-\frac{2}{5}$ &  $4_{\# 7 }^F \boxtimes 2_{\# 2}^B$ \\
        
        55 & $\frac{1}{20}$ & 49.410 & $1,1,\zeta^1_3\zeta^2_6,\zeta^1_3\zeta^2_6,\zeta^2_6,\zeta^2_6,\zeta^1_3,\zeta^1_3$ & $0,\frac{1}{2},\frac{3}{20},-\frac{7}{20},\frac{1}{4},-\frac{1}{4},-\frac{1}{10},\frac{2}{5}$ & $4_{\# 7 }^F  \boxtimes 2_{\#1}$ \\
        
        \rowcolor[gray]{0.9}
        56 & $-\frac{1}{20}$ & 8.478 & $1,1,-\frac{\zeta^1_3}{\zeta^2_6},-\frac{\zeta^1_3}{\zeta^2_6},-\frac{1}{\zeta^2_6},-\frac{1}{\zeta^2_6},\zeta^1_3,\zeta^1_3$ & $0,\frac{1}{2},\frac{3}{20},-\frac{7}{20},\frac{1}{4},-\frac{1}{4},-\frac{1}{10},\frac{2}{5}$ &  $4^F_{\#8} \boxtimes 2^B_{\# 1}$
        \\
        
        \rowcolor[gray]{0.9}
        57 & $-\frac{3}{20}$ & 18.873 & $1,1,-\frac{\zeta^2_6}{\zeta^1_3},-\frac{\zeta^2_6}{\zeta^1_3},\zeta^2_6,\zeta^2_6,-\frac{1}{\zeta^1_3},-\frac{1}{\zeta^1_3}$ & $0,\frac{1}{2},\frac{1}{20},-\frac{9}{20},\frac{1}{4},-\frac{1}{4},-\frac{1}{5},\frac{3}{10}$ & $4^F_{\#7} \boxtimes 2^{B}_{\# 4}$
        \\
        
        \rowcolor[gray]{0.9}
        58 & $-\frac{3}{20}$ & 3.2381 & $1,1,\frac{1}{\zeta^1_3\zeta^2_6},\frac{1}{\zeta^1_3\zeta^2_6},-\frac{1}{\zeta^2_6},-\frac{1}{\zeta^2_6},-\frac{1}{\zeta^1_3},-\frac{1}{\zeta^1_3}$ & $0,\frac{1}{2},\frac{1}{20},-\frac{9}{20},\frac{1}{4},-\frac{1}{4},-\frac{1}{5},\frac{3}{10}$ & $4^F_{\# 8} \boxtimes 2^{B}_{\# 4}$\\
        
        \rowcolor[gray]{0.9}
        59 & $\frac{3}{20}$ & 18.873 & $1,1,-\frac{\zeta^2_6}{\zeta^1_3},-\frac{\zeta^2_6}{\zeta^1_3},\zeta^2_6,\zeta^2_6,-\frac{1}{\zeta^1_3},-\frac{1}{\zeta^1_3}$ & $0,\frac{1}{2},-\frac{1}{20},\frac{9}{20},\frac{1}{4},-\frac{1}{4},\frac{1}{5},-\frac{3}{10}$ &  $4^F_{\# 7} \boxtimes 2^{B}_{\# 3}$\\
        
        \rowcolor[gray]{0.9}
        60 & $\frac{3}{20}$ & 3.2381 & $1,1,\frac{1}{\zeta^1_3\zeta^2_6},\frac{1}{\zeta^1_3\zeta^2_6},-\frac{1}{\zeta^2_6},-\frac{1}{\zeta^2_6},-\frac{1}{\zeta^1_3},-\frac{1}{\zeta^1_3}$ & $0,\frac{1}{2},-\frac{1}{20},\frac{9}{20},\frac{1}{4},-\frac{1}{4},\frac{1}{5},-\frac{3}{10}$ &  $4^F_{\# 8} \boxtimes 2^{B}_{\# 3}$\\
        
        \rowcolor[gray]{0.9}
        61 & $-\frac{1}{20}$ & 8.478 & $1,1,-\frac{\zeta^1_3}{\zeta^2_6},-\frac{\zeta^1_3}{\zeta^2_6},-\frac{1}{\zeta^2_6},-\frac{1}{\zeta^2_6},\zeta^1_3,\zeta^1_3$ & $0,\frac{1}{2},-\frac{3}{20},\frac{7}{20},\frac{1}{4},-\frac{1}{4},\frac{1}{10},-\frac{2}{5}$ &  $4^F_{\# 8} \boxtimes 2^{B}_{\# 2}$\\
        \hline
        
        62 & 0 & 93.254 & $1,1,\chi^1_2,\chi^1_2,\chi^1_2,\chi^1_2,\chi^3_8,\chi^3_8$ & $0,\frac{1}{2},\frac{1}{4},-\frac{1}{4},\frac{1}{4},-\frac{1}{4},0,\frac{1}{2}$ & $4_{\# 7}^F \boxtimes_{\cF_0} 4_{\# 7}^F $ \\
        
        \rowcolor[gray]{0.9}
        63 & 0 & 16 & $1,1,-1,-1,-\frac{1}{\chi^1_2},-\frac{1}{\chi^1_2},\chi^1_2,\chi^1_2$ & $0,\frac{1}{2},0,\frac{1}{2},\frac{1}{4},-\frac{1}{4},\frac{1}{4},-\frac{1}{4}$ & $4_{\# 7}^F \boxtimes_{\cF_0} 4_{\# 8}^F $\\
        
        \rowcolor[gray]{0.9}
        64 & 0 & 2.7452 & $1,1,\frac{1}{\qty(\chi^1_2)^2},\frac{1}{\qty(\chi^1_2)^2},-\frac{1}{\chi^1_2},-\frac{1}{\chi^1_2},-\frac{1}{\chi^1_2},-\frac{1}{\chi^1_2}$ & $0,\frac{1}{2},0,\frac{1}{2},\frac{1}{4},-\frac{1}{4},\frac{1}{4},-\frac{1}{4}$ & $4_{\# 8}^F \boxtimes_{\cF_0} 4_{\# 8}^F $\\
        \hline
        
        65 & $-\frac{1}{8}$ & 105.09 & $1,1,\zeta^2_{14},\zeta^2_{14},\zeta^4_{14},\zeta^4_{14},\zeta^6_{14},\zeta^6_{14}$ & $0,\frac{1}{2},-\frac{1}{8},\frac{3}{8},\frac{1}{8},-\frac{3}{8},\frac{1}{4},-\frac{1}{4}$ & Primitive  ($\cF_{(A_1,-14)}$) \\
        
        66 & $\frac{1}{8}$ & 105.09 & $1,1,\zeta^2_{14},\zeta^2_{14},\zeta^4_{14},\zeta^4_{14},\zeta^6_{14},\zeta^6_{14}$ & $0,\frac{1}{2},\frac{1}{8},-\frac{3}{8},-\frac{1}{8},\frac{3}{8},\frac{1}{4},-\frac{1}{4}$ & Primitive  ($\cF_{(A_1,14)}$) \\
        
        \rowcolor[gray]{0.9}
        67 & $\frac{1}{8}$ & 12.959 & $1,1,-\frac{\zeta^4_{14}}{\zeta^2_{14}},-\frac{\zeta^4_{14}}{\zeta^2_{14}},\frac{1}{\zeta^2_{14}},\frac{1}{\zeta^2_{14}},\frac{\zeta^6_{14}}{\zeta^2_{14}},\frac{\zeta^6_{14}}{\zeta^2_{14}}$ & $0,\frac{1}{2},\frac{1}{4},-\frac{1}{4},-\frac{1}{8},\frac{3}{8},\frac{1}{8},-\frac{3}{8}$ & Primitive \\
        
        \rowcolor[gray]{0.9}
        68 & $-\frac{1}{8}$ & 12.959 & $1,1,-\frac{\zeta^4_{14}}{\zeta^2_{14}},-\frac{\zeta^4_{14}}{\zeta^2_{14}},\frac{1}{\zeta^2_{14}},\frac{1}{\zeta^2_{14}},\frac{\zeta^6_{14}}{\zeta^2_{14}},\frac{\zeta^6_{14}}{\zeta^2_{14}}$ & $0,\frac{1}{2},\frac{1}{4},-\frac{1}{4},\frac{1}{8},-\frac{3}{8},-\frac{1}{8},\frac{3}{8}$ & Primitive \\
        
        \rowcolor[gray]{0.9}
        69 & $\frac{1}{8}$ & 5.7859 & $1,1,-\frac{\zeta^6_{14}}{\zeta^4_{14}},-\frac{\zeta^6_{14}}{\zeta^4_{14}},\frac{1}{\zeta^4_{14}},\frac{1}{\zeta^4_{14}},\frac{\zeta^2_{14}}{\zeta^4_{14}},\frac{\zeta^2_{14}}{\zeta^4_{14}}$ & $0,\frac{1}{2},-\frac{1}{8},\frac{3}{8},\frac{1}{8},-\frac{3}{8},\frac{1}{4},-\frac{1}{4}$ & Primitive \\
        
        \rowcolor[gray]{0.9}
        70 & $-\frac{1}{8}$ & 5.7859 & $1,1,-\frac{\zeta^6_{14}}{\zeta^4_{14}},-\frac{\zeta^6_{14}}{\zeta^4_{14}},\frac{1}{\zeta^4_{14}},\frac{1}{\zeta^4_{14}},\frac{\zeta^2_{14}}{\zeta^4_{14}},\frac{\zeta^2_{14}}{\zeta^4_{14}}$ & $0,\frac{1}{2},\frac{1}{8},-\frac{3}{8},-\frac{1}{8},\frac{3}{8},\frac{1}{4},-\frac{1}{4}$ &Primitive  \\
        
        \rowcolor[gray]{0.9}
        71 & $-\frac{1}{8}$ & 4.1583 & $1,1,-\frac{\zeta^4_{14}}{\zeta^6_{14}},-\frac{\zeta^4_{14}}{\zeta^6_{14}},-\frac{1}{\zeta^6_{14}},-\frac{1}{\zeta^6_{14}},\frac{\zeta^2_{14}}{\zeta^6_{14}},\frac{\zeta^2_{14}}{\zeta^6_{14}}$ & $0,\frac{1}{2},-\frac{1}{8},\frac{3}{8},\frac{1}{4},-\frac{1}{4},\frac{1}{8},-\frac{3}{8}$ & Primitive \\
        
        \rowcolor[gray]{0.9}
        72 & $\frac{1}{8}$ & 4.1583 & $1,1,-\frac{\zeta^4_{14}}{\zeta^6_{14}},-\frac{\zeta^4_{14}}{\zeta^6_{14}},-\frac{1}{\zeta^6_{14}},-\frac{1}{\zeta^6_{14}},\frac{\zeta^2_{14}}{\zeta^6_{14}},\frac{\zeta^2_{14}}{\zeta^6_{14}}$ & $0,\frac{1}{2},\frac{1}{8},-\frac{3}{8},\frac{1}{4},-\frac{1}{4},-\frac{1}{8},\frac{3}{8}$ &Primitive  \\
        \hline
    \end{tabular}}
    \caption{List of rank 8 fermionic MD. \emph{Continued.}}
    \label{tab:rank8b}
\end{table}


\begin{table}[t]
    \centering
    \resizebox{\columnwidth}{!}{
    \begin{tabular}{|c|c|c|c|c|c|}
        \hline
        \#& $c$ & $D^2$ & Quantum dimensions & Topological spins & Comments \\
        \hline
        1& 0 & 10 & $1,1,1,1,1,1,1,1,1,1$ & $0,\frac{1}{2},\frac{1}{10},-\frac{2}{5},\frac{1}{10},-\frac{2}{5},\frac{2}{5},-\frac{1}{10},\frac{2}{5},-\frac{1}{10}$ & $\cF_0 \boxtimes 5_4^B $ \\
        
        2& 0 & 10 & $1,1,1,1,1,1,1,1,1,1$ & $0,\frac{1}{2},\frac{1}{5},-\frac{3}{10},\frac{1}{5},-\frac{3}{10},\frac{3}{10},-\frac{1}{5},\frac{3}{10},-\frac{1}{5}$ & $\cF_0 \boxtimes 5_0^B$ \\
        
        \hline
        3& 0 & 24 & $1,1,1,1,\sqrt{3},\sqrt{3},\sqrt{3},\sqrt{3},2,2$ & $0,\frac{1}{2},0,\frac{1}{2},0,\frac{1}{2},0,\frac{1}{2},\frac{1}{6},-\frac{1}{3}$ &  
        Primitive\\
        
        4& 0 & 24 & $1,1,1,1,\sqrt{3},\sqrt{3},\sqrt{3},\sqrt{3},2,2$ & $0,\frac{1}{2},0,\frac{1}{2},0,\frac{1}{2},0,\frac{1}{2},\frac{1}{3},-\frac{1}{6}$ &  Primitive\\
        
        \rowcolor[gray]{0.9}
        5& 0 & 24 & $1,1,1,1,-\sqrt{3},-\sqrt{3},-\sqrt{3},-\sqrt{3},2,2$ & $0,\frac{1}{2},0,\frac{1}{2},0,\frac{1}{2},0,\frac{1}{2},\frac{1}{6},-\frac{1}{3}$ &  Primitive\\

        \rowcolor[gray]{0.9}
        6& 0 & 24 & $1,1,1,1,-\sqrt{3},-\sqrt{3},-\sqrt{3},-\sqrt{3},2,2$ & $0,\frac{1}{2},0,\frac{1}{2},0,\frac{1}{2},0,\frac{1}{2},\frac{1}{3},-\frac{1}{6}$ &  Primitive\\
        
        7& 0 & 24 & $1,1,1,1,\sqrt{3},\sqrt{3},\sqrt{3},\sqrt{3},2,2$ & $0,\frac{1}{2},0,\frac{1}{2},\frac{1}{4},-\frac{1}{4},\frac{1}{4},-\frac{1}{4},\frac{1}{6},-\frac{1}{3}$ &  Primitive\\
       
        8 & 0 & 24 & $1,1,1,1,\sqrt{3},\sqrt{3},\sqrt{3},\sqrt{3},2,2$ & $0,\frac{1}{2},0,\frac{1}{2},\frac{1}{4},-\frac{1}{4},\frac{1}{4},-\frac{1}{4},\frac{1}{3},-\frac{1}{6}$ &  Primitive\\
        
        \rowcolor[gray]{0.9}
        9& 0 & 24 & $1,1,1,1,-\sqrt{3},-\sqrt{3},-\sqrt{3},-\sqrt{3},2,2$ & $0,\frac{1}{2},0,\frac{1}{2},\frac{1}{4},-\frac{1}{4},\frac{1}{4},-\frac{1}{4},\frac{1}{6},-\frac{1}{3}$ & Primitive \\
        
        \rowcolor[gray]{0.9}
        10& 0 & 24 & $1,1,1,1,-\sqrt{3},-\sqrt{3},-\sqrt{3},-\sqrt{3},2,2$ & $0,\frac{1}{2},0,\frac{1}{2},\frac{1}{4},-\frac{1}{4},\frac{1}{4},-\frac{1}{4},\frac{1}{3},-\frac{1}{6}$ & Primitive \\
        
        11& 0 & 24 & $1,1,1,1,2,2,\sqrt{3},\sqrt{3},\sqrt{3},\sqrt{3}$ & $0,\frac{1}{2},0,\frac{1}{2},\frac{1}{6},-\frac{1}{3},\frac{1}{8},-\frac{3}{8},\frac{1}{8},-\frac{3}{8}$ & $\cF_0 \boxtimes 5^B_{\# 5}$ \\
      
        12 & 0 & 24 & $1,1,1,1,2,2,\sqrt{3},\sqrt{3},\sqrt{3},\sqrt{3}$ & $0,\frac{1}{2},0,\frac{1}{2},\frac{1}{6},-\frac{1}{3},\frac{3}{8},-\frac{1}{8},\frac{3}{8},-\frac{1}{8}$ & $\cF_0 \boxtimes 5^B_{\# 2}$ \\
        
        13 & 0 & 24 & $1,1,1,1,2,2,\sqrt{3},\sqrt{3},\sqrt{3},\sqrt{3}$ & $0,\frac{1}{2},0,\frac{1}{2},\frac{1}{3},-\frac{1}{6},\frac{1}{8},-\frac{3}{8},\frac{1}{8},-\frac{3}{8}$ & $\cF_0 \boxtimes 5^B_{\# 1}$ \\
        
        14 & 0 & 24 & $1,1,1,1,2,2,\sqrt{3},\sqrt{3},\sqrt{3},\sqrt{3}$ & $0,\frac{1}{2},0,\frac{1}{2},\frac{1}{3},-\frac{1}{6},\frac{3}{8},-\frac{1}{8},\frac{3}{8},-\frac{1}{8}$ & $\cF_0 \boxtimes 5^B_{\# 6}$ \\
        
        \rowcolor[gray]{0.9}
        15& 0 & 24 & $1,1,1,1,2,2,-\sqrt{3},-\sqrt{3},-\sqrt{3},-\sqrt{3}$ & $0,\frac{1}{2},0,\frac{1}{2},\frac{1}{6},-\frac{1}{3},\frac{1}{8},-\frac{3}{8},\frac{1}{8},-\frac{3}{8}$ & $\cF_0 \boxtimes 5^B_{\# 4}$ \\
        
        \rowcolor[gray]{0.9}
        16&0 & 24 & $1,1,1,1,2,2,-\sqrt{3},-\sqrt{3},-\sqrt{3},-\sqrt{3}$ & $0,\frac{1}{2},0,\frac{1}{2},\frac{1}{6},-\frac{1}{3},\frac{3}{8},-\frac{1}{8},\frac{3}{8},-\frac{1}{8}$ & $\cF_0 \boxtimes 5^B_{\# 8}$ \\
        
        \rowcolor[gray]{0.9}
        17& 0 & 24 & $1,1,1,1,2,2,-\sqrt{3},-\sqrt{3},-\sqrt{3},-\sqrt{3}$ & $0,\frac{1}{2},0,\frac{1}{2},\frac{1}{3},-\frac{1}{6},\frac{1}{8},-\frac{3}{8},\frac{1}{8},-\frac{3}{8}$ & $\cF_0 \boxtimes 5^B_{\# 7}$ \\
        
        \rowcolor[gray]{0.9}
        18& 0 & 24 & $1,1,1,1,2,2,-\sqrt{3},-\sqrt{3},-\sqrt{3},-\sqrt{3}$ & $0,\frac{1}{2},0,\frac{1}{2},\frac{1}{3},-\frac{1}{6},\frac{3}{8},-\frac{1}{8},\frac{3}{8},-\frac{1}{8}$ & $\cF_0 \boxtimes 5^B_{\# 3}$ \\
        
        \hline
        19& 0 & 40 & $1,1,1,1,2,2,2,2,\sqrt{10},\sqrt{10}$ & $0,\frac{1}{2},0,\frac{1}{2},\frac{1}{10},-\frac{2}{5},\frac{2}{5},-\frac{1}{10},\frac{1}{16},-\frac{7}{16}$ & Primitive ( $\cF_{U(1)_10/\ZZ_2}$)\\
        
        20& 0 & 40 & $1,1,1,1,2,2,2,2,\sqrt{10},\sqrt{10}$ & $0,\frac{1}{2},0,\frac{1}{2},\frac{1}{10},-\frac{2}{5},\frac{2}{5},-\frac{1}{10},\frac{3}{16},-\frac{5}{16}$ & Primitive \\
        
        21& 0 & 40 & $1,1,1,1,2,2,2,2,\sqrt{10},\sqrt{10}$ & $0,\frac{1}{2},0,\frac{1}{2},\frac{1}{10},-\frac{2}{5},\frac{2}{5},-\frac{1}{10},\frac{5}{16},-\frac{3}{16}$ &  Primitive\\
        
        22& 0 & 40 & $1,1,1,1,2,2,2,2,\sqrt{10},\sqrt{10}$ & $0,\frac{1}{2},0,\frac{1}{2},\frac{1}{10},-\frac{2}{5},\frac{2}{5},-\frac{1}{10},\frac{7}{16},-\frac{1}{16}$ & Primitive \\
        
        23& 0 & 40 & $1,1,1,1,2,2,2,2,\sqrt{10},\sqrt{10}$ & $0,\frac{1}{2},0,\frac{1}{2},\frac{1}{5},-\frac{3}{10},\frac{3}{10},-\frac{1}{5},\frac{1}{16},-\frac{7}{16}$ & Primitive \\
        
        24& 0 & 40 & $1,1,1,1,2,2,2,2,\sqrt{10},\sqrt{10}$ & $0,\frac{1}{2},0,\frac{1}{2},\frac{1}{5},-\frac{3}{10},\frac{3}{10},-\frac{1}{5},\frac{3}{16},-\frac{5}{16}$ & Primitive \\
        
        25& 0 & 40 & $1,1,1,1,2,2,2,2,\sqrt{10},\sqrt{10}$ & $0,\frac{1}{2},0,\frac{1}{2},\frac{1}{5},-\frac{3}{10},\frac{3}{10},-\frac{1}{5},\frac{5}{16},-\frac{3}{16}$ & Primitive \\
        
        26& 0 & 40 & $1,1,1,1,2,2,2,2,\sqrt{10},\sqrt{10}$ & $0,\frac{1}{2},0,\frac{1}{2},\frac{1}{5},-\frac{3}{10},\frac{3}{10},-\frac{1}{5},\frac{7}{16},-\frac{1}{16}$ & Primitive \\
        
        \rowcolor[gray]{0.9}
        27& 0 & 40 & $1,1,1,1,2,2,2,2,-\sqrt{10},-\sqrt{10}$ & $0,\frac{1}{2},0,\frac{1}{2},\frac{1}{10},-\frac{2}{5},\frac{2}{5},-\frac{1}{10},\frac{1}{16},-\frac{7}{16}$ & Primitive \\
        
        \rowcolor[gray]{0.9}
        28& 0 & 40 & $1,1,1,1,2,2,2,2,-\sqrt{10},-\sqrt{10}$ & $0,\frac{1}{2},0,\frac{1}{2},\frac{1}{10},-\frac{2}{5},\frac{2}{5},-\frac{1}{10},\frac{3}{16},-\frac{5}{16}$ & Primitive \\
        
        \rowcolor[gray]{0.9}
        29& 0 & 40 & $1,1,1,1,2,2,2,2,-\sqrt{10},-\sqrt{10}$ & $0,\frac{1}{2},0,\frac{1}{2},\frac{1}{10},-\frac{2}{5},\frac{2}{5},-\frac{1}{10},\frac{5}{16},-\frac{3}{16}$ & Primitive \\

        \rowcolor[gray]{0.9}
        30& 0 & 40 & $1,1,1,1,2,2,2,2,-\sqrt{10},-\sqrt{10}$ & $0,\frac{1}{2},0,\frac{1}{2},\frac{1}{10},-\frac{2}{5},\frac{2}{5},-\frac{1}{10},\frac{7}{16},-\frac{1}{16}$ & Primitive \\
        
        \rowcolor[gray]{0.9}
        31& 0 & 40 & $1,1,1,1,2,2,2,2,-\sqrt{10},-\sqrt{10}$ & $0,\frac{1}{2},0,\frac{1}{2},\frac{1}{5},-\frac{3}{10},\frac{3}{10},-\frac{1}{5},\frac{1}{16},-\frac{7}{16}$ &  Primitive\\
        
        \rowcolor[gray]{0.9}
        32& 0 & 40 & $1,1,1,1,2,2,2,2,-\sqrt{10},-\sqrt{10}$ & $0,\frac{1}{2},0,\frac{1}{2},\frac{1}{5},-\frac{3}{10},\frac{3}{10},-\frac{1}{5},\frac{3}{16},-\frac{5}{16}$ & Primitive \\
        
        \rowcolor[gray]{0.9}
        33& 0 & 40 & $1,1,1,1,2,2,2,2,-\sqrt{10},-\sqrt{10}$ & $0,\frac{1}{2},0,\frac{1}{2},\frac{1}{5},-\frac{3}{10},\frac{3}{10},-\frac{1}{5},\frac{5}{16},-\frac{3}{16}$ & Primitive \\
        
        \rowcolor[gray]{0.9}
        34& 0 & 40 & $1,1,1,1,2,2,2,2,-\sqrt{10},-\sqrt{10}$ & $0,\frac{1}{2},0,\frac{1}{2},\frac{1}{5},-\frac{3}{10},\frac{3}{10},-\frac{1}{5},\frac{7}{16},-\frac{1}{16}$ & Primitive \\
        
        \hline
        35 & $\frac{1}{22}$ & 69.2929 & $1,1,\zeta^1_9,\zeta^1_9,\zeta^2_9,\zeta^2_9,\zeta^3_9,\zeta^3_9,\zeta^4_9,\zeta^4_9$ & $0,\frac{1}{2},\frac{2}{11},-\frac{7}{22},\frac{7}{22},-\frac{2}{11},\frac{9}{22},-\frac{1}{11},\frac{5}{11},-\frac{1}{22}$ & $\cF_0 \boxtimes 5^B_{\# 15}$ \\
        
        36 & $-\frac{1}{22}$ & 69.2929 & $1,1,\zeta^4_9,\zeta^4_9,\zeta^3_9,\zeta^3_9,\zeta^2_9,\zeta^2_9,\zeta^1_9,\zeta^1_9$ & $0,\frac{1}{2},\frac{1}{22},-\frac{5}{11},\frac{1}{11},-\frac{9}{22},\frac{2}{11},-\frac{7}{22},\frac{7}{22},-\frac{2}{11}$ &  $\cF_0 \boxtimes 5^B_{\# 16}$\\
        
        \rowcolor[gray]{0.9}
        37& $\frac{5}{22}$ & 5.6137 & $1,1,-\frac{\zeta^1_{9}}{\zeta^4_{9}},-\frac{\zeta^1_{9}}{\zeta^4_{9}},-\frac{\zeta^3_{9}}{\zeta^4_{9}},-\frac{\zeta^3_{9}}{\zeta^4_{9}},\frac{\zeta^2_{9}}{\zeta^4_{9}},\frac{\zeta^2_{9}}{\zeta^4_{9}},\frac{1}{\zeta^4_{9}},\frac{1}{\zeta^4_{9}}$ & $0,\frac{1}{2},\frac{1}{22},-\frac{5}{11},\frac{1}{11},-\frac{9}{22},\frac{3}{11},-\frac{5}{22},\frac{9}{22},-\frac{1}{11}$ & $\cF_0 \boxtimes 5^B_{\# 23}$ \\
        
        \rowcolor[gray]{0.9}
        38& $-\frac{5}{22}$ & 5.6137 & $1,1,\frac{1}{\zeta^4_{9}},\frac{1}{\zeta^4_{9}},\frac{\zeta^2_{9}}{\zeta^4_{9}},\frac{\zeta^2_{9}}{\zeta^4_{9}},-\frac{\zeta^3_{9}}{\zeta^4_{9}},-\frac{\zeta^3_{9}}{\zeta^4_{9}},-\frac{\zeta^1_{9}}{\zeta^4_{9}},-\frac{\zeta^1_{9}}{\zeta^4_{9}}$ & $0,\frac{1}{2},\frac{1}{11},-\frac{9}{22},\frac{5}{22},-\frac{3}{11},\frac{9}{22},-\frac{1}{11},\frac{5}{11},-\frac{1}{22}$ & $\cF_0 \boxtimes 5^B_{\# 18}$ \\
        
        \rowcolor[gray]{0.9}
        39& $\frac{2}{11}$ & 6.64709 & $1,1,\frac{\zeta^4_{9}}{\zeta^3_{9}},\frac{\zeta^4_{9}}{\zeta^3_{9}},-\frac{\zeta^2_{9}}{\zeta^3_{9}},-\frac{\zeta^2_{9}}{\zeta^3_{9}},-\frac{1}{\zeta^3_{9}},-\frac{1}{\zeta^3_{9}},-\frac{\zeta^1_{9}}{\zeta^3_{9}},-\frac{\zeta^1_{9}}{\zeta^3_{9}}$ & $0,\frac{1}{2},\frac{3}{22},-\frac{4}{11},\frac{5}{22},-\frac{3}{11},\frac{3}{11},-\frac{5}{22},\frac{7}{22},-\frac{2}{11}$ & $\cF_0 \boxtimes 5^B_{\# 22}$ \\
        
        \rowcolor[gray]{0.9}
        40& $-\frac{2}{11}$ & 6.64709 & $1,1,-\frac{\zeta^1_{9}}{\zeta^3_{9}},-\frac{\zeta^1_{9}}{\zeta^3_{9}},-\frac{1}{\zeta^3_{9}},-\frac{1}{\zeta^3_{9}},-\frac{\zeta^2_{9}}{\zeta^3_{9}},-\frac{\zeta^2_{9}}{\zeta^3_{9}},\frac{\zeta^4_{9}}{\zeta^3_{9}},\frac{\zeta^4_{9}}{\zeta^3_{9}}$ & $0,\frac{1}{2},\frac{2}{11},-\frac{7}{22},\frac{5}{22},-\frac{3}{11},\frac{3}{11},-\frac{5}{22},\frac{4}{11},-\frac{3}{22}$ & $\cF_0 \boxtimes 5^B_{\# 19}$ \\
        
        \rowcolor[gray]{0.9}
        41& $\frac{3}{22}$ & 9.62957 & $1,1,\frac{\zeta^4_{9}}{\zeta^2_{9}},\frac{\zeta^4_{9}}{\zeta^2_{9}},-\frac{1}{\zeta^2_{9}},-\frac{1}{\zeta^2_{9}},-\frac{\zeta^3_{9}}{\zeta^2_{9}},-\frac{\zeta^3_{9}}{\zeta^2_{9}},\frac{\zeta^1_{9}}{\zeta^2_{9}},\frac{\zeta^1_{9}}{\zeta^2_{9}}$ & $0,\frac{1}{2},\frac{1}{22},-\frac{5}{11},\frac{5}{22},-\frac{3}{11},\frac{4}{11},-\frac{3}{22},\frac{5}{11},-\frac{1}{22}$ &  $\cF_0 \boxtimes 5^B_{\# 20}$\\
        
        \rowcolor[gray]{0.9}
        42& $-\frac{3}{22}$ & 9.62957 & $1,1,\frac{\zeta^1_{9}}{\zeta^2_{9}},\frac{\zeta^1_{9}}{\zeta^2_{9}},-\frac{\zeta^3_{9}}{\zeta^2_{9}},-\frac{\zeta^3_{9}}{\zeta^2_{9}},-\frac{1}{\zeta^2_{9}},-\frac{1}{\zeta^2_{9}},\frac{\zeta^4_{9}}{\zeta^2_{9}},\frac{\zeta^4_{9}}{\zeta^2_{9}}$ & $0,\frac{1}{2},\frac{1}{22},-\frac{5}{11},\frac{3}{22},-\frac{4}{11},\frac{3}{11},-\frac{5}{22},\frac{5}{11},-\frac{1}{22}$ &  $\cF_0 \boxtimes 5^B_{\# 21}$\\
        
        \rowcolor[gray]{0.9}
        43& $\frac{1}{11}$ & 18.8168 & $1,1,\frac{\zeta^4_{9}}{\zeta^1_{9}},\frac{\zeta^4_{9}}{\zeta^1_{9}},-\frac{\zeta^2_{9}}{\zeta^1_{9}},-\frac{\zeta^2_{9}}{\zeta^1_{9}},-\frac{\zeta^3_{9}}{\zeta^1_{9}},-\frac{\zeta^3_{9}}{\zeta^1_{9}},\frac{1}{\zeta^1_{9}},\frac{1}{\zeta^1_{9}}$ & $0,\frac{1}{2},\frac{3}{22},-\frac{4}{11},\frac{7}{22},-\frac{2}{11},\frac{4}{11},-\frac{3}{22},\frac{9}{22},-\frac{1}{11}$ & $\cF_0 \boxtimes 5^B_{\# 24}$ \\
        
        \rowcolor[gray]{0.9} 
        44& $-\frac{1}{11}$ & 18.8168 & $1,1,\frac{1}{\zeta^1_{9}},\frac{1}{\zeta^1_{9}},-\frac{\zeta^3_{9}}{\zeta^1_{9}},-\frac{\zeta^3_{9}}{\zeta^1_{9}},-\frac{\zeta^2_{9}}{\zeta^1_{9}},-\frac{\zeta^2_{9}}{\zeta^1_{9}},\frac{\zeta^4_{9}}{\zeta^1_{9}},\frac{\zeta^4_{9}}{\zeta^1_{9}}$ & $0,\frac{1}{2},\frac{1}{11},-\frac{9}{22},\frac{3}{22},-\frac{4}{11},\frac{2}{11},-\frac{7}{22},\frac{4}{11},-\frac{3}{22}$ &  $\cF_0 \boxtimes 5^B_{\# 17}$ \\
        \hline 
    \end{tabular}}
    \caption{List of rank 10 fermionic MD.}
    \label{tab:rank10}
\end{table}

\begin{table}[t]
    \centering
    \resizebox{\columnwidth}{!}{
    \begin{tabular}{|c|c|c|c|c|c|}
        \hline
        \# & $c$ & $D^2$ & Quantum dimensions & Topological spins & Comments \\
        
        \hline
        45& $\frac{1}{14}$ & 70.6848 & $1,1,\zeta^2_{12},\zeta^2_{12},\zeta^2_5,\zeta^2_5,\zeta^2_5,\zeta^2_5,\zeta^4_{12},\zeta^4_{12}$ & $0,\frac{1}{2},\frac{1}{7},-\frac{5}{14},\frac{5}{14},-\frac{1}{7},\frac{5}{14},-\frac{1}{7},\frac{3}{7},-\frac{1}{14}$ &  $\cF_0 \boxtimes 5^B_{\# 10}$\\
        
        46& $-\frac{1}{14}$ & 70.6848 & $1,1,\zeta^4_{12},\zeta^4_{12},\zeta^2_5,\zeta^2_5,\zeta^2_5,\zeta^2_5,\zeta^2_{12},\zeta^2_{12}$ & $0,\frac{1}{2},\frac{1}{14},-\frac{3}{7},\frac{1}{7},-\frac{5}{14},\frac{1}{7},-\frac{5}{14},\frac{5}{14},-\frac{1}{7}$ & $\cF_0 \boxtimes 5^B_{\# 9}$ \\
        
        \rowcolor[gray]{0.9}
        47& $\frac{1}{7}$ & 4.3117 & $1,1,\frac{\zeta^2_5}{\zeta^4_{12}},\frac{\zeta^2_5}{\zeta^4_{12}},\frac{\zeta^2_5}{\zeta^4_{12}},\frac{\zeta^2_5}{\zeta^4_{12}},-\frac{1}{\zeta^4_{12}},-\frac{1}{\zeta^4_{12}},-\frac{\zeta^2_{12}}{\zeta^4_{12}},-\frac{\zeta^2_{12}}{\zeta^4_{12}}$ & $0,\frac{1}{2},\frac{3}{14},-\frac{2}{7},\frac{3}{14},-\frac{2}{7},\frac{2}{7},-\frac{3}{14},\frac{5}{14},-\frac{1}{7}$ &  $\cF_0 \boxtimes 5^B_{\# 14}$\\
        
        \rowcolor[gray]{0.9} 
        48& $-\frac{1}{7}$ & 4.3117 & $1,1,-\frac{\zeta^2_{12}}{\zeta^4_{12}},-\frac{\zeta^2_{12}}{\zeta^4_{12}},-\frac{1}{\zeta^4_{12}},-\frac{1}{\zeta^4_{12}},\frac{\zeta^2_5}{\zeta^4_{12}},\frac{\zeta^2_5}{\zeta^4_{12}},\frac{\zeta^2_5}{\zeta^4_{12}},\frac{\zeta^2_5}{\zeta^4_{12}}$ & $0,\frac{1}{2},\frac{1}{7},-\frac{5}{14},\frac{3}{14},-\frac{2}{7},\frac{2}{7},-\frac{3}{14},\frac{2}{7},-\frac{3}{14}$ &  $\cF_0 \boxtimes 5^B_{\# 11}$\\
        
        \rowcolor[gray]{0.9} 
        49& $\frac{3}{14}$ & 9.00346 & $1,1,-\frac{\zeta^2_{5}}{\zeta^2_{12}},-\frac{\zeta^2_{5}}{\zeta^2_{12}},-\frac{\zeta^2_{5}}{\zeta^2_{12}},-\frac{\zeta^2_{5}}{\zeta^2_{12}},-\frac{1}{\zeta^2_{12}},-\frac{1}{\zeta^2_{12}},\frac{\zeta^4_{12}}{\zeta^2_{12}},\frac{\zeta^4_{12}}{\zeta^2_{12}}$ & $0,\frac{1}{2},\frac{1}{14},-\frac{3}{7},\frac{1}{14},-\frac{3}{7},\frac{2}{7},-\frac{3}{14},\frac{3}{7},-\frac{1}{14}$ &  $\cF_0 \boxtimes 5^B_{\# 12}$\\
        
        \rowcolor[gray]{0.9} 
        50& $-\frac{3}{14}$ & 9.00346 & $1,1,\frac{\zeta^4_{12}}{\zeta^2_{12}},\frac{\zeta^4_{12}}{\zeta^2_{12}},-\frac{1}{\zeta^2_{12}},-\frac{1}{\zeta^2_{12}},-\frac{\zeta^2_{5}}{\zeta^2_{12}},-\frac{\zeta^2_{5}}{\zeta^2_{12}},-\frac{\zeta^2_{5}}{\zeta^2_{12}},-\frac{\zeta^2_{5}}{\zeta^2_{12}}$ & $0,\frac{1}{2},\frac{1}{14},-\frac{3}{7},\frac{3}{14},-\frac{2}{7},\frac{3}{7},-\frac{1}{14},\frac{3}{7},-\frac{1}{14}$ & $\cF_0 \boxtimes 5^B_{\# 13}$ \\
        
        \hline
        51& $\frac{1}{5}$ & 204.317 & $1,1,\zeta^8_{18},\zeta^8_{18},\zeta^6_{18},\zeta^6_{18},\zeta^2_{18},\zeta^2_{18},\zeta^4_{18},\zeta^4_{18}$ & $0,\frac{1}{2},0,\frac{1}{2},\frac{1}{10},-\frac{2}{5},\frac{1}{10},-\frac{2}{5},\frac{3}{10},-\frac{1}{5}$ &  Primitive ($\cF_{(A_1, 18)} $ )\\
        
        52& $-\frac{1}{5}$ & 204.317 & $1,1,\zeta^8_{18},\zeta^8_{18},\zeta^4_{18},\zeta^4_{18},\zeta^6_{18},\zeta^6_{18},\zeta^2_{18},\zeta^2_{18}$ & $0,\frac{1}{2},0,\frac{1}{2},\frac{1}{5},-\frac{3}{10},\frac{2}{5},-\frac{1}{10},\frac{2}{5},-\frac{1}{10}$ &  Primitive ($\cF_{(A_1, -18)} $ ) \\
        
        \rowcolor[gray]{0.9}
        53& $\frac{1}{5}$ & 5.12543 & $1,1,\frac{1}{\zeta^8_{18}},\frac{1}{\zeta^8_{18}},-\frac{\zeta^2_{18}}{\zeta^8_{18}},-\frac{\zeta^2_{18}}{\zeta^8_{18}},-\frac{\zeta^6_{18}}{\zeta^8_{18}},-\frac{\zeta^6_{18}}{\zeta^8_{18}},\frac{\zeta^4_{18}}{\zeta^8_{18}},\frac{\zeta^4_{18}}{\zeta^8_{18}}$ & $0,\frac{1}{2},0,\frac{1}{2},\frac{1}{10},-\frac{2}{5},\frac{1}{10},-\frac{2}{5},\frac{3}{10},-\frac{1}{5}$ & Primitive \\
        
        \rowcolor[gray]{0.9}
        54& $-\frac{1}{5}$ & 5.12543 & $1,1,\frac{1}{\zeta^8_{18}},\frac{1}{\zeta^8_{18}},\frac{\zeta^4_{18}}{\zeta^8_{18}},\frac{\zeta^4_{18}}{\zeta^8_{18}},-\frac{\zeta^2_{18}}{\zeta^8_{18}},-\frac{\zeta^2_{18}}{\zeta^8_{18}},-\frac{\zeta^6_{18}}{\zeta^8_{18}},-\frac{\zeta^6_{18}}{\zeta^8_{18}}$ & $0,\frac{1}{2},0,\frac{1}{2},\frac{1}{5},-\frac{3}{10},\frac{2}{5},-\frac{1}{10},\frac{2}{5},-\frac{1}{10}$ & Primitive \\
        
        \rowcolor[gray]{0.9}
        55& $\frac{1}{10}$ & 6.29808 & $1,1,-\frac{\zeta^2_{18}}{\zeta^6_{18}},-\frac{\zeta^2_{18}}{\zeta^6_{18}},\frac{\zeta^8_{18}}{\zeta^6_{18}},\frac{\zeta^8_{18}}{\zeta^6_{18}},-\frac{1}{\zeta^6_{18}},-\frac{1}{\zeta^6_{18}},-\frac{\zeta^4_{18}}{\zeta^6_{18}},-\frac{\zeta^4_{18}}{\zeta^6_{18}}$ & $0,\frac{1}{2},0,\frac{1}{2},\frac{3}{10},-\frac{1}{5},\frac{3}{10},-\frac{1}{5},\frac{2}{5},-\frac{1}{10}$ & Primitive \\
        
        \rowcolor[gray]{0.9}
        56&$-\frac{1}{10}$ & 6.29808 & $1,1,-\frac{\zeta^2_{18}}{\zeta^6_{18}},-\frac{\zeta^2_{18}}{\zeta^6_{18}},-\frac{\zeta^4_{18}}{\zeta^6_{18}},-\frac{\zeta^4_{18}}{\zeta^6_{18}},\frac{\zeta^8_{18}}{\zeta^6_{18}},\frac{\zeta^8_{18}}{\zeta^6_{18}},-\frac{1}{\zeta^6_{18}},-\frac{1}{\zeta^6_{18}}$ & $0,\frac{1}{2},0,\frac{1}{2},\frac{1}{10},-\frac{2}{5},\frac{1}{5},-\frac{3}{10},\frac{1}{5},-\frac{3}{10}$ & Primitive \\
        
        \rowcolor[gray]{0.9} 
        57& $\frac{1}{10}$ & 24.2592 & $1,1,-\frac{\zeta^6_{18}}{\zeta^2_{18}},-\frac{\zeta^6_{18}}{\zeta^2_{18}},-\frac{1}{\zeta^2_{18}},-\frac{1}{\zeta^2_{18}},\frac{\zeta^8_{18}}{\zeta^2_{18}},\frac{\zeta^8_{18}}{\zeta^2_{18}},\frac{\zeta^4_{18}}{\zeta^2_{18}},\frac{\zeta^4_{18}}{\zeta^2_{18}}$ & $0,\frac{1}{2},0,\frac{1}{2},\frac{3}{10},-\frac{1}{5},\frac{3}{10},-\frac{1}{5},\frac{2}{5},-\frac{1}{10}$ & Primitive \\
        
        \rowcolor[gray]{0.9}
        58& $-\frac{1}{10}$ & 24.2592 & $1,1,-\frac{\zeta^6_{18}}{\zeta^2_{18}},-\frac{\zeta^6_{18}}{\zeta^2_{18}},\frac{\zeta^4_{18}}{\zeta^2_{18}},\frac{\zeta^4_{18}}{\zeta^2_{18}},-\frac{1}{\zeta^2_{18}},-\frac{1}{\zeta^2_{18}},\frac{\zeta^8_{18}}{\zeta^2_{18}},\frac{\zeta^8_{18}}{\zeta^2_{18}}$ & $0,\frac{1}{2},0,\frac{1}{2},\frac{1}{10},-\frac{2}{5},\frac{1}{5},-\frac{3}{10},\frac{1}{5},-\frac{3}{10}$ & Primitive \\
        
        \hline
        59& 0 & 472.379 & $1,1,\chi^4_{15},\chi^4_{15},\chi^5_{15},\chi^5_{15},\chi^3_{15},\chi^3_{15},\chi^3_{15},\chi^3_{15}$ & $0,\frac{1}{2},0,\frac{1}{2},\frac{1}{6},-\frac{1}{3},\frac{1}{10},-\frac{2}{5},\frac{2}{5},-\frac{1}{10}$ &  Primitive, $\hat{N} \leq 3$\\
        
        60& 0 & 472.379 & $1,1,\chi^4_{15},\chi^4_{15},\chi^5_{15},\chi^5_{15},\chi^3_{15},\chi^3_{15},\chi^3_{15},\chi^3_{15}$ & $0,\frac{1}{2},0,\frac{1}{2},\frac{1}{6},-\frac{1}{3},\frac{1}{5},-\frac{3}{10},\frac{3}{10},-\frac{1}{5}$ & Primitive, $\hat{N} \leq 3$ \\
        
        61& 0 & 472.379 & $1,1,\chi^4_{15},\chi^4_{15},\chi^5_{15},\chi^5_{15},\chi^3_{15},\chi^3_{15},\chi^3_{15},\chi^3_{15}$ & $0,\frac{1}{2},0,\frac{1}{2},\frac{1}{3},-\frac{1}{6},\frac{1}{10},-\frac{2}{5},\frac{2}{5},-\frac{1}{10}$ &  Primitive, $\hat{N} \leq 3$\\
        
        62& 0 & 472.379 & $1,1,\chi^4_{15},\chi^4_{15},\chi^5_{15},\chi^5_{15},\chi^3_{15},\chi^3_{15},\chi^3_{15},\chi^3_{15}$ & $0,\frac{1}{2},0,\frac{1}{2},\frac{1}{3},-\frac{1}{6},\frac{1}{5},-\frac{3}{10},\frac{3}{10},-\frac{1}{5}$ &  Primitive, $\hat{N} \leq 3$\\
        
        \rowcolor[gray]{0.9}
        63& 0 & 7.621 & $1,1,\frac{1}{\chi^4_{15}},\frac{1}{\chi^4_{15}},\frac{\chi^5_{15}}{\chi^4_{15}},\frac{\chi^5_{15}}{\chi^4_{15}},-\frac{\chi^3_{15}}{\chi^4_{15}},-\frac{\chi^3_{15}}{\chi^4_{15}},-\frac{\chi^3_{15}}{\chi^4_{15}},-\frac{\chi^3_{15}}{\chi^4_{15}}$ & $0,\frac{1}{2},0,\frac{1}{2},\frac{1}{6},-\frac{1}{3},\frac{1}{10},-\frac{2}{5},\frac{2}{5},-\frac{1}{10}$ &  Primitive, $\hat{N} \leq 3$\\
        
        \rowcolor[gray]{0.9}
        64& 0 & 7.621 & $1,1,\frac{1}{\chi^4_{15}},\frac{1}{\chi^4_{15}},\frac{\chi^5_{15}}{\chi^4_{15}},\frac{\chi^5_{15}}{\chi^4_{15}},-\frac{\chi^3_{15}}{\chi^4_{15}},-\frac{\chi^3_{15}}{\chi^4_{15}},-\frac{\chi^3_{15}}{\chi^4_{15}},-\frac{\chi^3_{15}}{\chi^4_{15}}$ & $0,\frac{1}{2},0,\frac{1}{2},\frac{1}{6},-\frac{1}{3},\frac{1}{5},-\frac{3}{10},\frac{3}{10},-\frac{1}{5}$ &  Primitive, $\hat{N} \leq 3$\\
        
        \rowcolor[gray]{0.9} 
        65& 0 & 7.621 & $1,1,\frac{1}{\chi^4_{15}},\frac{1}{\chi^4_{15}},\frac{\chi^5_{15}}{\chi^4_{15}},\frac{\chi^5_{15}}{\chi^4_{15}},-\frac{\chi^3_{15}}{\chi^4_{15}},-\frac{\chi^3_{15}}{\chi^4_{15}},-\frac{\chi^3_{15}}{\chi^4_{15}},-\frac{\chi^3_{15}}{\chi^4_{15}}$ & $0,\frac{1}{2},0,\frac{1}{2},\frac{1}{3},-\frac{1}{6},\frac{1}{10},-\frac{2}{5},\frac{2}{5},-\frac{1}{10}$ &  Primitive, $\hat{N} \leq 3$\\
        
        \rowcolor[gray]{0.9}
        66& 0 & 7.621 & $1,1,\frac{1}{\chi^4_{15}},\frac{1}{\chi^4_{15}},\frac{\chi^5_{15}}{\chi^4_{15}},\frac{\chi^5_{15}}{\chi^4_{15}},-\frac{\chi^3_{15}}{\chi^4_{15}},-\frac{\chi^3_{15}}{\chi^4_{15}},-\frac{\chi^3_{15}}{\chi^4_{15}},-\frac{\chi^3_{15}}{\chi^4_{15}}$ & $0,\frac{1}{2},0,\frac{1}{2},\frac{1}{3},-\frac{1}{6},\frac{1}{5},-\frac{3}{10},\frac{3}{10},-\frac{1}{5}$ &  Primitive, $\hat{N} \leq 3$\\
        
        \hline
        67& 0 & 475.151 & $1,1,\chi^5_{24},\chi^5_{24},\chi^3_{6},\chi^3_{6},\chi^3_{6},\chi^3_{6},\chi^4_{24},\chi^4_{24}$ & $0,\frac{1}{2},0,\frac{1}{2},\frac{1}{4},-\frac{1}{4},\frac{1}{4},-\frac{1}{4},\frac{1}{6},-\frac{1}{3}$ &  Primitive, $\hat{N} \leq 4$\\
        
        68& 0 & 475.151 & $1,1,\chi^5_{24},\chi^5_{24},\chi^3_{6},\chi^3_{6},\chi^3_{6},\chi^3_{6},\chi^4_{24},\chi^4_{24}$ & $0,\frac{1}{2},0,\frac{1}{2},\frac{1}{4},-\frac{1}{4},\frac{1}{4},-\frac{1}{4},\frac{1}{3},-\frac{1}{6}$ &  Primitive, $\hat{N} \leq 4$\\
        
        \rowcolor[gray]{0.9}
        69& 0 & 4.84898 & $1,1,\frac{1}{\chi^5_{24}},\frac{1}{\chi^5_{24}},\frac{\chi^3_6}{\chi^5_{24}},\frac{\chi^3_6}{\chi^5_{24}},\frac{\chi^3_6}{\chi^5_{24}},\frac{\chi^3_6}{\chi^5_{24}},-\frac{\chi^4_{24}}{\chi^5_{24}},-\frac{\chi^4_{24}}{\chi^5_{24}}$ & $0,\frac{1}{2},0,\frac{1}{2},\frac{1}{4},-\frac{1}{4},\frac{1}{4},-\frac{1}{4},\frac{1}{6},-\frac{1}{3}$ &  Primitive, $\hat{N} \leq 4$\\
        
        \rowcolor[gray]{0.9}
        70& 0 & 4.84898 & $1,1,\frac{1}{\chi^5_{24}},\frac{1}{\chi^5_{24}},\frac{\chi^3_6}{\chi^5_{24}},\frac{\chi^3_6}{\chi^5_{24}},\frac{\chi^3_6}{\chi^5_{24}},\frac{\chi^3_6}{\chi^5_{24}},-\frac{\chi^4_{24}}{\chi^5_{24}},-\frac{\chi^4_{24}}{\chi^5_{24}}$ & $0,\frac{1}{2},0,\frac{1}{2},\frac{1}{4},-\frac{1}{4},\frac{1}{4},-\frac{1}{4},\frac{1}{3},-\frac{1}{6}$ & Primitive, $\hat{N} \leq 4$ \\
        \hline
    \end{tabular}}
    \caption{List of rank 10 fermionic MD. \emph{Continued}. MD of \#59--\#62 and \#67--\#68 are of new unitary SMCs which are previously unknown, and \#63--\#66 and \#69--\#70 are their nonunitary counterparts with the same fusion rules.}
    \label{tab:rank10b}
\end{table}

\end{document}